\documentclass{article}
\usepackage[utf8]{inputenc}
\usepackage[english]{babel}
\usepackage{ dsfont }
\usepackage{ stmaryrd }
\usepackage{url}
\usepackage{hyperref}
\usepackage[affil-it]{authblk}

\usepackage{amsmath,amsthm,amssymb}
\usepackage{todonotes}
\usepackage{bm}
\usepackage{enumerate}
\usepackage{textcomp}
\usepackage{array}
\usepackage{wrapfig}
\usepackage{multirow}
\usepackage{algorithm, algpseudocode}
\usepackage{algorithmicx}  
\usepackage{tabularx}
\usepackage{csquotes}
\usepackage[english]{babel}
\usepackage{graphicx}
\usepackage[margin=3.3cm]{geometry}
\usepackage{latexsym}
\usepackage{hyperref}
\usepackage{amssymb}
\usepackage{subfigure}
\DeclareUnicodeCharacter{2217}{*}

\DeclareMathOperator{\suc}{Success}

\DeclareMathOperator{\unif}{Unif}

\hypersetup{colorlinks=true,citecolor=blue,linkcolor=blue,filecolor=blue,urlcolor=blue}

\theoremstyle{plain} 
\newtheorem{theorem}{Theorem}[section]

\usepackage{url}
\newtheorem{lemma}[theorem]{Lemma}

\newtheorem{corollary}[theorem]{Corollary}
\newtheorem{proposition}[theorem]{Proposition}

\theoremstyle{definition}

\setuptodonotes{ color=green!40, shadow}

\newcommand*{\bN}{\mathbb{N}}

\newcommand*{\bid}{\mathbf{1}}

\newcommand*{\cD}{\mathcal{D}}

\newcommand*{\cE}{\mathcal{E}}

\newcommand*{\cI}{\mathcal{I}}
\newcommand*{\dI}{\mathbb{I}}
 
\newcommand*{\cN}{\mathcal{N}}
\newcommand*{\cM}{\mathcal{M}}

\newcommand*{\cS}{\mathcal{S}}
\newcommand*{\fS}{\mathfrak{S}}

\newcommand*{\cW}{\mathcal{W}}
\newcommand*{\cY}{\mathcal{Y}}
\newcommand*{\cX}{\mathcal{X}}

\newcommand*{\mge}{\succcurlyeq}
\newcommand*{\mle}{\preccurlyeq}
\newcommand{\pin}{p_{\rm{initial}}}
\newcommand{\pt}{p_{\rm{target}}}
\newcommand*{\pr}[1]{\mathbb{P}\left[#1 \right]}

\newcommand*{\exs}[2]{\mathbb{E}_{#1}\left[#2 \right]}

\DeclareMathOperator{\TV}{TV}

\newcommand*{\eps}{\varepsilon}

\newcommand*{\id}{\mathrm{id}}
\newcommand*{\tr}[1]{\mathrm{Tr}\left[#1\right]}
\newcommand*{\ptr}[2]{\mathrm{Tr}_{#1}\left[#2\right]}
\newcommand*{\ket}[1]{| #1 \rangle}
\newcommand*{\bra}[1]{\langle #1 |}
\newcommand*{\spr}[2]{\langle #1 | #2 \rangle}
\newcommand*{\proj}[1]{|#1\rangle\!\langle #1|}

\newcommand{\ceil}[1]	{\left\lceil #1 \right\rceil}

\usepackage[sorting=none,maxbibnames=99]{biblatex}

\bibliography{bib.bib}

\title{Optimality of meta-converse for channel simulation}

\author[1]{Aadil Oufkir}
\author[2]{Omar Fawzi}
\author[1]{Mario Berta} 

\affil[1]{\small{Institute for Quantum Information,
  RWTH Aachen University,
  Aachen, Germany}}
\affil[2]{\small{Univ Lyon, Inria, ENS Lyon, UCBL, LIP, Lyon, France}
}

\begin{document}

\maketitle

	\begin{abstract}
		We study the effect of shared non-signaling correlations for the problem of simulating a channel using noiseless communication in the one-shot setting. For classical channels, we show how to round any non-signaling-assisted simulation strategy\,---\,which corresponds to the natural linear programming meta-converse for channel simulation\,---\,to a strategy that only uses shared randomness. For quantum channels, we round any non-signaling-assisted simulation strategy to a strategy that only uses shared entanglement. Our main result is for classical and classical-quantum channels, for which we employ ideas from approximation algorithms to give a guarantee on the ratio of success probabilities of at least $(1-\mathrm{e}^{-1})$. We further show this ratio to be optimal for the purely classical case. It can be improved to $(1-t^{-1})$ using $O(\ln \ln(t))$ additional bits of communication.
	\end{abstract}


\section{Introduction}

\subsection{Motivation}

Channel simulation corresponds to the reverse task of channel coding and is about simulating a noisy channel using a noiseless one (see, e.g.,~\cite{Steinberg1994May,bennett2002entanglement,berta2013quantum,sudan2019communication,cao2022channel,yu2022common}). Similarly as for channel coding, in the asymptotic, independent and identically distributed (i.i.d.) limit\,---\,and when shared-randomness assistance is available\,---\,the channel's capacity characterizes the minimal communication rate needed for channel simulation with vanishing error~\cite{bennett2002entanglement,berta2011quantum,bennett2014quantum}.
	
Such asymptotic characterizations in Shannon theory~\cite{shannon1948mathematical} have since been much refined and for the most general (structureless) one-shot setting, e.g., the work~\cite{barman2017algorithmic} on classical channel coding gives a simple and efficient approximation algorithm to return a code achieving a $(1-\mathrm{e}^{-1})$-approximation of the maximum success probability that can be attained. This algorithm is based on a natural linear programming relaxation of the problem and can be understood as the sender and receiver sharing non-signaling correlations~\cite{matthews2012linear}\,---\,which exactly corresponds to the well-known PPV meta-converse for channel coding~\cite{polyanskiy2010channel,Hayashi09}.\footnote{The PPV meta-converse is often stated in terms of the code size $M$ for a fixed error probability $\eps$, while in the following we are interested in maximizing the success probability $1-\eps$ for a fixed given code size $M$.} The approximation ratio $(1-\mathrm{e}^{-1})$ is also tight and it is NP-hard to achieve a strictly better ratio~\cite{barman2017algorithmic}. Further, some extensions about computing the maximum success probability of (deterministic) broadcast channels~\cite{fawzi2023broadcast}, multiple access channels~\cite{fawzi2023multiple}, as well as classical-quantum (CQ) channels~\cite{fawzi2019approximation} are available. However, for quantum channel coding, finding efficient and provably optimal approximation algorithms for the optimal success probability remains open.

Are similar one-shot refinements possible for the reverse task of channel simulation? In this work, we consider in the one-shot setting the problem of simulating classical channels with shared-randomness assistance, as well as the problem of simulating quantum channels with shared-entanglement assistance. In what might be called the algorithmic point of view on Shannon theory~\cite{polyanskiy2010channel,matthews2012linear,barman2017algorithmic,fawzi2019approximation,fawzi2023multiple,fawzi2023broadcast}\,---\,and similar as in aforementioned results on channel coding\,---\,one has a complete description of the channel to simulate and the goal is then to find good encoding-decoding schemes in order to maximize the success probability for a fixed communication size. More specifically, we are interested in designing efficient approximation algorithms to obtain near optimal codes for channel simulation.


\subsection{Overview of results}

Our starting point is the natural linear programming meta-converse relaxation for channel simulation with shared-randomness assistance, which is known to correspond to non-signaling correlations assistance~\cite{cubitt2011zero,cao2022channel}. This relaxation then also gives the correct channel capacity formula in the asymptotic i.i.d.\ limit, as well as higher-order refinements thereof~\cite{cao2022channel}. Here, we show that this meta-converse for channel simulation is equally useful in the one-shot setting:

\begin{itemize}
	\item Following \cite{cubitt2011zero,cao2022channel}, we allow non-signaling correlations assistance between the sender and the receiver (cf.~\cite{leung_power_2015,Duan2015Dec,Wang2017Aug,fang2019quantum}). The optimal success probability for the simulation then becomes a linear program (LP) analogous to the meta-converse for channel coding~\cite{polyanskiy2010channel,matthews2012linear}.

	\item This computationally efficient LP meta-converse provides an upper bound on the success probability, as non-signaling (NS) correlations include shared-randomness (SR) assistance. Our main result is to show that this bound $\suc^{\rm{NS}}$ gives an $(1-\mathrm{e}^{-1})$-approximation of the maximum success probability $\suc^{\rm{SR}}$ as
    \[    \suc^{\rm{NS}} \ge \suc^{\rm{SR}}\ge \Big(1-\frac{1}{\mathrm{e}}\Big) \suc^{\rm{NS}}. \] 
    To derive this, we round the solution of the non-signaling program to a shared-randomness strategy for simulating a classical channel with the same communication size. After reducing the non-signaling program using symmetries~\cite{cubitt2011zero}, we use rejection sampling~\cite{neumann1951various} to simulate an optimal non-signaling channel using an optimal probability distribution that appears in the non-signaling program. Even further, we prove a meta inequality between the shared-randomness and non-signaling channels, allowing for an approximation flexible in distance measures.  Rejection sampling was used, in a different fashion, in communication~\cite{Harsha10}.

	\item We show that the approximation ratio $(1-\mathrm{e}^{-1})$ is tight. This follows by considering a type of universal channels~\cite{cubitt2011zero} that can be perfectly simulated using non-signaling strategies (with a sufficient communication power) while  shared-randomness strategies (with the same amount of  communication) can only achieve a success probability  bounded by $(1-\mathrm{e}^{-1})$.

    \item If an additional $\ln \ln(t)$ bits of communication is allowed, then we obtain a better approximation ratio $(1-t^{-1})$. In particular, this allows us to prove that the non-signaling simulation capacity of a classical channel is exactly the same as with only shared-randomness assistance~\cite{bennett2002entanglement}.

    \item We generalize our classical approximation results to the classical-quantum setting, where we compare shared-entanglement with non-signaling strategies. The coherent rejection sampling from \cite{anshu2017quantum} known as convex split technique does not seem to be good enough to achieve the multiplicative approximation factor $(1-\mathrm{e}^{-1})$. So, instead we proceed by a more refined quantum form of rejection sampling technique as described in~\cite{cao2023channel} for classical-quantum channels.

   \item Finally, we prove weaker approximation results for fully quantum channels using the coherent rejection sampling \cite{anshu2017quantum} along with the minimax theorem proven by \cite{Cao2024Mar}. Erratum: We note that in the preliminary version of our work \cite{berta2024optimality} presented at the 2024 IEEE International Symposium on Information Theory, we claimed stronger rounding results for fully quantum channels in \cite[Section III]{berta2024optimality}, similar to the bounds we have for the classical and classical-quantum case. However, in \cite{berta2024optimality} we did not use a proper fully quantum and worst case notion of simulating fully quantum channels. Namely, the proof implicitly assumes that the reference system is classical as well as the input to be fixed. Consequently, we do know if the claimed \cite[Proposition 2 \& Corollary 2.1]{berta2024optimality} hold and they remain open questions. As such, we would herewith like to retract \cite[Section III]{berta2024optimality}, and replace it with Section \ref{sec:quantum} in our full version here. Note that the results still hold for the classical-quantum case, as mentioned above.
\end{itemize}


\subsection{Notation} 
For a positive integer $n\in \bN$, we denote by $[n]$ the set of integers between $1$ and $n$.   The total variation $(\TV)$ distance between two probability distributions $p$ and $q$ on $[n]$ is 
	\[
		\|p-q\|_{\TV}=\frac{1}{2}\sum_{i=1}^n |p_i-q_i|.
	\]
A classical channel $W_{Y|X}$ with input alphabet $\cX$ and output alphabet $\cY$ is a collection $\{W_{Y|X}(\cdot|x)\}_{x\in \cX}$ of probability distributions on $\cY$.  The worst case $\TV$ distance between two classical channels $W_{Y|X}$ and $T_{Y|X}$ is defined as
\[\left\|W_{Y|X}-T_{Y|X}\right\|_{\TV} = \sup_{x\in \cX } \left\|W_{Y|X}(\cdot|x)-T_{Y|X}(\cdot|x)\right\|_{\TV}.\]
We use standard quantum information theory notation. 
Hilbert spaces are denoted $A, B,...$ and will be confused with the label of the corresponding systems. We let $|A|$ be the dimension of the Hilbert space $A$. Let $\mathrm{L}(A)$ denote the set of linear maps from $A$ to itself. A quantum state on $A$ is defined as
\[
\rho\in \mathrm{L}(A):\; \rho \mge0 \;\text{and}\; \tr{\rho}=1,
\]
where $\rho \mge 0$ means that $\rho$ is positive semidefinite. The set of quantum states on $A$ is denoted by $\mathrm{D}(A)$. 
For an integer $N\ge 2$, we denote the $N$-partite composite system by $A_1A_2\cdots A_N = A_1\otimes A_2\otimes \dots \otimes A_N$.
\\The partial trace $\ptr{B}{.}$ is a quantum channel from $AB$ to $A$ defined as 
\[
\forall \rho \in \mathrm{L}(A)\otimes \mathrm{L}(B) :	\ptr{B}{\rho}= \sum_{i=1}^{|B|} \left(\dI_{A} \otimes \bra{i}_B \right) \,\rho\, \left(\dI_{A} \otimes \ket{i}_B \right). 
\]
\\Quantum channels are linear maps $\cN: \mathrm{L}(A)\rightarrow\mathrm{L}(B)$ that can be written in the form 
\[
	\forall \rho \in \mathrm{L}(A) : \cN(\rho)= \sum_{x\in \cX} K_x \rho K_x^\dagger 
\]
where the Kraus operators $\{K_x\}_{x\in \cX}$ are linear maps from $A$ to $B$ and satisfy $\sum_{x\in \cX}  K_x^\dagger K_x= \dI_{A}$.
\\For a quantum channel $\cN_{A\rightarrow B}$, we define the Choi matrix \[J_{\cN}=\sum_{i,j=1}^{|A|} \ket{i}\bra{j}_{A'}\otimes \cN_{A\rightarrow B}(\ket{i}\bra{j}_A) = (\id\otimes \cN_{A\rightarrow B})(\proj{w}_{A'A}) \quad \text{where}\quad  \ket{w}=\sum_{i=1}^{|A|}\ket{i}_{A'}\otimes \ket{i}_{A}. \] 
A quantum channel is completely positive and trace preserving:
\begin{itemize}
    \item \textbf{completely positive:} for all $n\in \bN$ and $\rho \mge 0$, $[\id_n\otimes \cN](\rho)\mge 0$ or equivalently $J_\cN\mge 0$,
    \item \textbf{trace preserving:} for all $\rho$, $\tr{\cN(\rho)} = \tr{\rho}$ or equivalently $\ptr{B}{J_\cN} = \dI_{A'}$. 
\end{itemize}
A quantum channel $\cN$ with classical input system is called a classical-quantum (CQ) channel, and is described by a set of quantum states $\{\cN(x)\}_{x \in \cX} \in \left(\mathrm{D}(B)\right)^{\cX}$. A measurement is described by a POVM (positive operator-valued measure) $\{\Lambda_x\}_{x \in \cX} \in \left(\mathrm{L}(A)\right)^{\cX}$ where the measurement operators satisfy $\Lambda_x \mge 0$ and $\sum_{x\in \cX} \Lambda_x= \dI_{A}$. After performing the measurement on a quantum state $\rho\in \mathrm{D}(A)$ we observe the outcome $x\in \cX$ with probability $\tr{\Lambda_x\rho}$.
\\The one norm of an operator $\xi \in \mathrm{L}(A)$ is defined as $\|\xi\|_1= \tr{|\xi|}$ where $|\xi|=\sqrt{\xi \xi^\dagger}$. 
The trace norm between two quantum states $\rho\in \mathrm{D}(A)$ and $\sigma\in \mathrm{D}(A)$ is defined as \[\|\rho-\sigma\|_{\operatorname{Tr}}=\tfrac{1}{2}\|\rho-\sigma\|_1.\] 
The diamond distance between two quantum channels $\cN_{A\rightarrow B}$ and $\cM_{A\rightarrow B}$ is defined as 
 \[\|\cN-\cM\|_{\diamond}=\sup_{\sigma\in \mathrm{D}(RA)}\|(\id_R\otimes\cN_{A\rightarrow B})(\sigma_{RA})-(\id_R\otimes\cM_{A\rightarrow B})(\sigma_{RA})\|_{\operatorname{Tr}}.\]
The diamond distance between two CQ channels $\{\cN(x)\}_{x \in \cX}, \{\cM(x)\}_{x \in \cX} \in \left(\mathrm{D}(B)\right)^{\cX}$  
satisfies
 \[\|\cN-\cM\|_{\diamond}=\sup_{x\in \cX}\|\cN(x)-\cM(x)\|_{\operatorname{Tr}}.\]
The fidelity between two quantum states 
$\rho\in \mathrm{D}(A)$ and $\sigma\in \mathrm{D}(A)$ is defined as \[F(\rho,\sigma) = \|\sqrt{\rho}\sqrt{\sigma}\|_1^2.\] 
The purified distance  between two quantum states 
$\rho\in \mathrm{D}(A)$ and $\sigma\in \mathrm{D}(A)$ is defined as \[P(\rho,\sigma) = \sqrt{1-F(\rho, \sigma)}.\]
The purified distance between two quantum channels $\cN_{A\rightarrow B}$ and $\cM_{A\rightarrow B}$ is defined as 
 \[P(\cN,\cM)=\sup_{\sigma\in \mathrm{D}(RA)}P\left((\id_R\otimes\cN_{A\rightarrow B})(\sigma_{RA}),(\id_R\otimes\cM_{A\rightarrow B})(\sigma_{RA})\right).\]
The relative entropy between two quantum states 
$\rho\in \mathrm{D}(A)$ and $\sigma\in \mathrm{D}(A)$ such that $\mathrm{supp}(\rho)\subset \mathrm{supp}(\sigma)$ is defined as \[D(\rho,\sigma) = \tr{\rho(\log(\rho)-\log(\sigma))}.\] 
Finally, $A\mge B$ stands for $A-B$  positive semi-definite.
	

\section{Simulation of classical channels}\label{sec:classical}

\subsection{Shared-randomness assistance}

Given a classical channel $W_{Y|X}$ of input alphabet $\cX$ and output alphabet $\cY$ and an integer $M$, our goal is to simulate the channel $W$ using a classical communication of at most $M$ distinct messages and with an error probability as small as possible. More formally, we can describe a size $M$ $\eps$-simulation code  for the channel  $W_{Y|X}$ by a triple $(\{\cE_s \}_{s\in \cS}, \{\cD_s \}_{s\in \cS}, \{p_S(s)\}_{s\in \cS} )$ such that the synthesized channel
\begin{align*}
	\widetilde{W}_{Y|X}(y|x)= \sum_{s\in \cS} p_S(s) \sum_{i=1}^M \cE_s(i|x)\cD_s(y|i)
\end{align*}
is $\eps$ close to the actual channel $W_{Y|X}$ in the worst case total variation distance. Here, it is allowed to use a shared random variable $S$ on an arbitrary discrete set $\cS$. This choice is motivated by the fact that strategies assisted by shared-randomness (SR) are proven to achieve the optimal channel simulation capacity \cite{cao2022channel}, whereas strategies without assistance require a simulation rate strictly greater than the capacity~\cite{bennett2002entanglement,bennett2014quantum}. In the following, we describe a possible approximation algorithm given by the meta-converse.


\subsection{Non-signaling assistance in the classical setting}\label{sec:cl-ns}
 
SR strategies are special case of the so called non-signaling (NS) strategies. In this latter, the joint encoder-decoder map $N_{IY|XJ}$ satisfies
\begin{align*}
	N_{IY|XJ}(i,y|x,j)&\ge 0  && \forall i,y,x,j,
	\\ \sum_{i,y}    N_{IY|XJ}(i,y|x,j)&=1  && \forall x,j,
	\\    \sum_{i}    N_{IY|XJ}(i,y|x,j)&=N_{Y|J}(y|j)  && \forall x, 
	\\    \sum_{y}    N_{IY|XJ}(i,y|x,j)&=N_{I|X}(i|x)  && \forall j.
	\end{align*}
We denote the set of non-signaling maps by $\cN\cS(IY|XJ)$. Non-signaling strategies prove to be useful for simplifying the computation of the maximal success probability which is obtained by solving the following program
\begin{align}\label{classical-ns-program}
	\suc^{\rm{NS}}(W,M) = \max_{N} &\;1-\sup_{x\in \cX} \left\| \widetilde{W}(\cdot|x)- W(\cdot|x)\right\|_{\TV} \notag
	\\  \rm{s.t.}\;\;& N\in \cN\cS(IY|XJ),
   \\ &  \widetilde{W}_{}(y|x)= \hspace{-0.3em}\sum_{i=1}^M N_{}(i,y|x,i). \notag
\end{align}
So, when we relax the constraints of the encoder-decoder to be non-signaling, we obtain a linear program. After a symmetry-based reduction, the program~\eqref{classical-ns-program} becomes~\cite{fang2019quantum,cao2022channel} (see Lemma~\ref{app:ns-program-cc} for a proof)
\begin{align}\label{classical-ns-program-sym}
	\suc^{\rm{NS}}(W,M) = \max_{\widetilde{W},\, \zeta} \;\;& 1-\sup_{x\in \cX} \left\| \widetilde{W}(\cdot|x)- W(\cdot|x)\right\|_{\TV} \notag 
	\\ \rm{s.t.}\;\;&  \sum_{y} \widetilde{W}_{}(y|x)= 1 \quad \forall x, \notag
	\\&   \widetilde{W}_{}(y|x)\ge 0 \quad \quad \;\;\, \forall x,y, 
	\\&  \widetilde{W}_{}(y|x) \le \zeta(y)  \quad\;\, \forall x,y, \notag
	\\&  \sum_{y} \zeta(y)=M.  \notag
\end{align}
Observe that a strategy assisted by shared-randomness $N(iy|xj)= \sum_{s\in \cS} p_S(s) \sum_{i=1}^M \cE_s(i|x)\cD_s(y|j)$ is non-signaling, so the non-signaling (NS) simulation success probability is greater than the shared-randomness (SR) simulation success probability
\begin{align*}
	\suc^{\rm{NS}}(W_{},M) \ge  \suc^{\rm{SR}}(W_{},M).
\end{align*}
These values might differ in general, even when comparing them with entanglement-assisted strategies (see Appendix~\ref{app: one shot SR vs EA}). Since non-signaling strategies are computationally accessible and shared-randomness strategies are of interest, a natural question arises:

\begin{center}
    \emph{How large is the gap between the success probabilities \\ of shared-randomness and non-signaling strategies?} 
\end{center}


\subsection{Rounding}

In the following, we round the solution of the non-signaling program \eqref{classical-ns-program-sym} to a strategy that requires only shared-randomness. This allows us to prove a meta inequality between the channel using non-signaling (NS) resources and the constructed channel using only shared randomness (SR).

\begin{proposition}\label{prop:meta-inequality-classical}
	Let $M,M'\in \bN$ and $\widetilde{W}^{\rm{NS}}_{}$ be a feasible solution of the program \eqref{classical-ns-program-sym} of size $M$. Then, there exists a shared-randomness assisted strategy $\widetilde{W}^{\rm{SR}}_{}$ of size $M'$ such that
	\begin{align*}
		\widetilde{W}^{\rm{SR}}_{}(y|x)\ge \left[1-\left(1-\frac{1}{M}\right)^{M'}\right]\widetilde{W}^{\rm{NS}}_{}(y|x)\quad\forall x,y.
	\end{align*}
\end{proposition}

To construct the strategy $\widetilde{W}^{\rm{SR}}_{}$ with shared-randomness, we use a standard tool from statistics \cite{neumann1951various,Robert} that has been applied previously to channel simulation \cite{cao2022channel,Cao22,cao2023channel} namely the \emph{rejection sampling technique} \cite{Jain2003Jun,Harsha10}. However, the one shot bound we obtain in this work (Proposition \ref{prop:meta-inequality-classical}) is conceptually different and highly flexible.

\begin{proof}[Proof of Proposition~\ref{prop:meta-inequality-classical}]
Let $(\widetilde{W}^{\rm{NS}}_{}, \zeta)$ be a feasible solution of the program \eqref{classical-ns-program-sym}.

	\paragraph{Shared randomness} Since $\sum_y \zeta(y)=M$ and $\zeta(y)\ge 0$ for all $y$,  $ \{\frac{\zeta(y)}{M}\}_y$ is a probability distribution. Let $\mathbf{Y_1, \dots, Y_{M'}}$ be i.i.d.\ samples from $\pin= \{\frac{\zeta(y)}{M}\}_y$, which acts as the shared-randomness assistance of the protocol $\widetilde{W}^{\rm{SR}}_{}(y|x)$ to be constructed.

	\paragraph{Encoding} For an input $x$, we run the rejection sampling algorithm \cite{cao2022channel} with $M'$ steps for $\pin= \{\frac{\zeta(y)}{M}\}_y$ and $\pt = \widetilde{W}^{\rm{NS}}_{}(\cdot|x) $ and obtain $\mathbf{\widetilde{Y}}= \mathbf{Y_{i}}$ where $\mathbf{i}$  is the first index (or $M'$ if it does not exist) such that
	\begin{align*}
		\mathbf{U_i} \le  \frac{1}{M}\cdot \frac{\pt(\mathbf{Y_i})}{\pin(\mathbf{Y_i})}= \frac{\widetilde{W}^{\rm{NS}}(\mathbf{Y_i}|x)}{\zeta(\mathbf{Y_i})},
	\end{align*}
	where $\mathbf{U_i} \sim\unif ([0,1])$, then encode $\cE_{s}(m|x)=\delta_{m=\mathbf{i}}$.

	\paragraph{Decoding} Decode as $\cD_s(y|m)= \delta_{y=\mathbf{Y_m}}$. For finite $M'$, we follow \cite{cao2022channel} and show, with a simple calculation,  that the distribution of $\mathbf{\widetilde{Y}}$ satisfies for $y\in \cY$ that
	\begin{align*}
		p_{\mathbf{\widetilde{Y}}}(y)&=   \left(1-\frac{1}{M}\right)^{M'-1} \left(\pin(y)-\frac{1}{M}\cdot \pt(y)\right)
		 + \left[1-\left(1-\frac{1}{M}\right)^{M'}\right]\pt(y).
	\end{align*}
	Indeed, let us denote $\mathbf{E}$ the binary random variable such that $\mathbf{E}=1$ if and only if there is an $\mathbf{i}\le M' $ such that $ \mathbf{U_i} \le  \frac{1}{M}\cdot \frac{\pt(\mathbf{Y_i})}{\pin(\mathbf{Y_i})}$. We have by letting $\lambda=\frac{1}{M}$ that
	\begin{align*}     p_{\mathbf{\widetilde{Y}E}}(y,0) &= \sum_{y_1, \dots, y_{M'-1} } \prod_{i=1}^{M'-1}\pin(y_i)\left(1-\lambda\cdot \frac{\pt(y_i)}{\pin(y_i)}\right)
		\cdot \pin(y)\left(1-\lambda\cdot \frac{\pt(y)}{\pin(y)}\right)   
		\\&= \prod_{i=1}^{M'-1} \sum_{y_i}\left(\pin(y_i)-\lambda\cdot \pt(y_i)\right)    \cdot \left(\pin(y)-\lambda\cdot \pt(y)\right) 
		\\&= (1-\lambda)^{M'-1} \left(\pin(y)-\lambda\cdot \pt(y)\right). 
	\end{align*}
	Similarly, we get
		\begin{align*}     p_{\mathbf{\widetilde{Y}E}}(y,1) 
			&= \sum_{j=1}^{M'}\sum_{y_1, \dots, y_{j-1} } \prod_{i=1}^{j-1}\pin(y_i)\left(1-\lambda\cdot \frac{\pt(y_i)}{\pin(y_i)}\right)
			 \cdot \pin(y) \left(\lambda\cdot \frac{\pt(y)}{\pin(y)}\right)  
			\\&= \sum_{j=1}^{M'}(1-\lambda)^{j-1}\cdot \lambda\cdot  \pt(y)
			\\&= \left(1-(1-\lambda)^{M'}\right) \pt(y).
	\end{align*}
	Hence, we find
	\begin{align*}
		p_{\mathbf{\widetilde{Y}}}(y)&=  p_{\mathbf{\widetilde{Y}E}}(y,0) +p_{\mathbf{\widetilde{Y}E}}(y,1) 
		\\&=\left(1-\frac{1}{M}\right)^{M'-1} \left(\pin(y)-\frac{1}{M}\cdot \pt(y)\right)
		+ \left[1-\left(1-\frac{1}{M}\right)^{M'}\right]\pt(y). 
	\end{align*}
	Moreover, the second constraint of the program~\eqref{classical-ns-program-sym} reads
	\begin{align*}
			\forall y:\;    \pin(y)= \frac{\zeta(y)}{M}\ge \frac{1}{M}\widetilde{W}^{\rm{NS}}(y|x)=\frac{1}{M}\cdot \pt(y).
	\end{align*} 
	Therefore, we get
	\begin{align*}
		p_{\mathbf{\widetilde{Y}}}(y)&=\left(1-\frac{1}{M}\right)^{M'-1} \left(\pin(y)-\frac{1}{M}\cdot \pt(y)\right)
		+ \left[1-\left(1-\frac{1}{M}\right)^{M'}\right]\pt(y) 
		\\&\ge \left[1-\left(1-\frac{1}{M}\right)^{M'}\right]\pt(y) .
	\end{align*}
	By observing that $p_{\mathbf{\widetilde{Y}}}(y)= \widetilde{W}^{\rm{SR}}(y|x)$, where $\widetilde{W}^{\rm{SR}}(y|x)= \sum_{s\in \cS} p_S(s) \sum_{m=1}^{M'} \cE_s(m|x)\cD_s(y|m)$, we deduce the required inequality. 
\end{proof}

As a direct corollary of Proposition~\ref{prop:meta-inequality-classical}, we can control the gap between the success probabilities of strategies assisted by shared-randomness and non-signaling assistance. Although we focus on the worst-case total variation distance, it is noteworthy that Proposition~\ref{prop:meta-inequality-classical} also permits to control the gap between the success probabilities of SR and NS strategies under average-case total variation, as well as under average and worst-case Bhattacharyya distance.

\begin{corollary}\label{cor: succ-classical}
	Let $M,M'\ge 1$, and $W$ be a channel. Then, we have
	\begin{align*}
		& 1\ge\frac{\suc^{\rm{SR}}(W_{}, M')  }{\suc^{\rm{NS}}(W_{}, M)} \ge  \left[1-\left(1-\frac{1}{M}\right)^{M'}\right]. 
	\end{align*}
\end{corollary}

\begin{proof}[Proof of Corollary~\ref{cor: succ-classical}]
Observe that for two probability distributions $p $ and $q$, $1- \|p-q\|_{\TV} =\sum_i \min (p_i, q_i)$. Let $x\in \cX$, by Proposition~\ref{prop:meta-inequality-classical} we have
\begin{align*}
	1-\left\|\widetilde{W}^{\rm{SR}}_{}(\cdot |x) - W_{}(\cdot |x)  \right\|_{\TV} 
	&= \sum_y \min\left(\widetilde{W}^{\rm{SR}}_{}(y |x) , \; W_{}(y |x)    \right)
	\\&\ge  \sum_y \left[1-\left(1-\frac{1}{M}\right)^{M'}\right]  \min\left(\widetilde{W}^{\rm{NS}}_{}(y |x) , \; W_{}(y |x)    \right)
	\\&=  \left[1-\left(1-\frac{1}{M}\right)^{M'}\right] \left( 1-\left\|\widetilde{W}^{\rm{NS}}_{}(\cdot |x) - W_{}(\cdot |x)  \right\|_{\TV}\right),
\end{align*}
where we used $\left(1-(1-1/M)^{M'}\right) \le 1$. Choosing the optimal feasible solution $\widetilde{W}^{\rm{NS}}_{}$ of the program~\eqref{classical-ns-program-sym} and taking the minimum on $x$ on both sides give the desired inequality. 
\end{proof}
In particular, when $M'=M$, the gap between $\suc^{\rm{SR}}(W_{}, M) $ and $\suc^{\rm{NS}}(W_{}, M) $ is at most $1-\frac{1}{\mathrm{e}}$. Furthermore, by choosing $M'=\ln(t)M$, we can show that the gap is at most $1-\frac{1}{t}$ and thus approaches $1$ as $t\rightarrow \infty$. This implies the fact that the non-signaling assistance does not help to reduce the asymptotic simulation capacity of a classical channel~\cite{bennett2014quantum,cao2022channel}. In words, the simulation capacity of a classical channel is the asymptotic minimum rate of communication required to simulate this channel with high success probability in the presence of shared randomness (SR), shared entanglement (EA), or non-signaling correlations (NS).

\begin{corollary}\label{app:class-cap-equal}
	Let $W$ be a channel. Then, we have
	\begin{align*}
		C^{\rm{SR}}(W)=C^{\rm{EA}}(W) =C^{\rm{NS}}(W).
	\end{align*}
\end{corollary}

\begin{proof}
Since NS contains EA and EA contains SR, we have that $C^{\rm{SR}}(W)\ge C^{\rm{EA}}(W) \ge C^{\rm{NS}}(W)$ so it suffices to show that $C^{\rm{SR}}(W)=C^{\rm{NS}}(W)$.
\\From Corollary~\ref{cor: succ-classical} we have for all $n\in \bN$, $M=2^t$ and $M'=\ln(t) 2^t$ that
\begin{align*}
	\suc^{\rm{SR}}(W^{\otimes n}, \ln(t) 2^t)  
	&\ge  \left[1-\left(1-\frac{1}{2^t}\right)^{\ln(t) 2^t}\right]\suc^{\rm{NS}}(W^{\otimes n}, 2^t) 
	\\& \ge \left(1-\frac{1}{t}\right)\suc^{\rm{NS}}(W^{\otimes n}, 2^t). 
\end{align*}
Suppose by contradiction that there is $R$ such that 
\begin{align*}
	C^{\rm{NS}}(W)<R<C^{\rm{SR}}(W).
\end{align*}
So, by taking $t=Rn$ we have 
\begin{align*}
	\suc^{\rm{NS}}(W^{\otimes n}, 2^{Rn})\underset{n \rightarrow \infty}{\longrightarrow}1
\end{align*}
and as $1-\frac{1}{t}= 1-\frac{1}{Rn} \underset{n \rightarrow \infty}{\longrightarrow}1$ we deduce that 
\begin{align*}
	\suc^{\rm{SR}}(W^{\otimes n}, \ln(Rn)2^{Rn})\underset{n \rightarrow \infty}{\longrightarrow}1,
\end{align*}
hence we find that $R=\lim_{n \rightarrow \infty}\frac{1}{n}\log(\ln(Rn)2^{Rn})\ge C^{\rm{SR}}(W)$, which contradicts the assumption $R<C^{\rm{SR}}(W)$.
\end{proof}


\subsection{Tightness}\label{sec:tightness}

Finally, we show that the bounds of Corollary~\ref{cor: succ-classical} are in fact tight.

\begin{lemma}\label{app:optimal}
For integers $n \geq k$ and $M = \ceil{\frac{n}{k}}$, consider the channel
\begin{align*}
   W : \binom{n}{k}&\rightarrow \{1,2, \dots, n\},\;W(y|x)=\frac{1}{k}\;\mathbf{1}\{y \in x\}.
\end{align*}
Then, we have that
\begin{align*}
    \text{$\suc^{\rm{NS}}(W,M)=1$, while $\suc^{\rm{SR}}(W,M')\le 1-\frac{1}{\binom{n}{k}}\binom{n-M'}{k}$.}
\end{align*}
\end{lemma}

This shows that the ratio in Corollary~\ref{cor: succ-classical} cannot be improved in general. In fact, let $M'=k$ and $M=\gamma k$ for some constant $\gamma$. Then, we immediately have
\begin{align*}
    \frac{\binom{kM-M'}{k}}{\binom{kM}{k}} &= \frac{(kM-M')\cdots (kM-M'-k+1)}{(kM)\cdots (kM-k+1)}
  \\ & \geq \left(1-\frac{k}{\gamma k^2 - k + 1}\right)^k\underset{k \to \infty}{\sim}\left(1-\frac{1}{\gamma k}\right)^{k}\underset{k \to \infty}{\longrightarrow} \mathrm{e}^{-1/\gamma}.
\end{align*}
In particular, for the channel~\cite{cubitt2011zero}
\begin{align}\label{eq:example-channel}
W : \binom{M^2}{M}\rightarrow \{1,2, \dots, M^2\}, \; W(y|x)=\frac{1}{M}\;\mathbf{1}\{y \in x\}
\end{align}
we have that
\begin{align*}
\frac{\suc^{\rm{SR}}(W_{},M)}{\suc^{\rm{NS}}(W_{},M)}\underset{M\rightarrow \infty}{\longrightarrow} 1-\frac{1}{\mathrm{e}}.
\end{align*}

\begin{proof}[Proof of Lemma~\ref{app:optimal}]
For the first equality, we can choose $\zeta(y)= \frac{1}{k}$ for all $y\in [n]$, hence $\widetilde{W}=W$ satisfies the constraints of the program \eqref{classical-ns-program-sym}. For the second inequality, the proof is similar to an argument in~\cite{cubitt2011zero}. We use $1-\|p-q\|_{\TV}= \sum_{i}\min(p_i,q_i)$ to upper bound the success probability as
\begin{align*}
	\suc^{\rm{SR}}(W,M')
	&= \inf_x \sum_{y} \min\left(\widetilde{W}^{\rm{SR}}(y|x),  \frac{1}{k}\;\mathbf{1}\{y \in x\}\right)
	\\&= \inf_x \sum_{y\in x} \min\left(\widetilde{W}^{\rm{SR}}(y|x),  \frac{1}{k}\right)
	\\&\le \inf_x \sum_{y\in x} \widetilde{W}^{\rm{SR}}(y|x)
	\\&\le \frac{1}{\binom{n}{k}}\sum_{x\in \binom{n}{k}} \sum_{y\in x} \sum_{s\in \cS} p_S(s) \sum_{m=1}^{M'} \cE_s(m|x)\cD_s(y|m)
	\\&=\sum_{s\in \cS} p_S(s) \frac{1}{\binom{n}{k}}\sum_{x\in \binom{n}{k}} \sum_{y\in x}  \sum_{m=1}^{M'} \cE_s(m|x)\cD_s(y|m).
\end{align*}
We may suppose without loss of generality that $\cD_s$ is deterministic (any stochastic map  can be written as a convex combination of deterministic maps \cite{Davis}: $\cD_s= \sum_{i\in \cS_s}p_s(i) \cD_{i,s}$ where $\cD_{i,s}$ is deterministic, then we can consider the shared randomness $\{p_S(s)p_s(i)\}_{s\in \cS,i\in \cS_s}$). So, for a fixed $s\in \cS$, the set $\cY_s = \{y : \cD_s(y|m) \neq 0 \text{ for some } m\}$ is of size at most $M'$. Hence, we find
\begin{align*}
	 \suc^{\rm{SR}}(W,M')
	&\le\sum_{s\in \cS} p_S(s) \frac{1}{\binom{n}{k}}\sum_{x\in \binom{n}{k}} \sum_{y\in x}  \sum_{m=1}^{M'} \cE_s(m|x)\cD_s(y|m)
	\\&= \sum_{s\in \cS} p_S(s) \frac{1}{\binom{n}{k}}\sum_{x\in \binom{n}{k}: x\cap \cY_s= \emptyset } \sum_{y\in x}  \sum_{m=1}^{M'} \cE_s(m|x)\cdot 0
	\\&\quad+ \sum_{s\in \cS} p_S(s) \frac{1}{\binom{n}{k}}\sum_{x\in \binom{n}{k}: x\cap \cY_s\neq \emptyset} \sum_{y\in x}  \sum_{m=1}^{M'} \cE_s(m|x)\cD_s(y|m)
	\\&\le \sum_{s\in \cS} p_S(s) \frac{1}{\binom{n}{k}}\sum_{x\in \binom{n}{k}: x\cap \cY_s\neq \emptyset} 1
	\\&= \sum_{s\in \cS} p_S(s) \frac{1}{\binom{n}{k}}\left(\binom{n}{k}- \binom{n-M'}{k}\right)
	\\&=\frac{1}{\binom{n}{k}}\left(\binom{n}{k}- \binom{n-M'}{k}\right).
\end{align*}
So, the gap between the NS and SR success probabilities can be as large as
\begin{align*}
    \frac{\suc^{\rm{SR}}(W,M')}{\suc^{\rm{NS}}(W,M)}\le 1-\frac{\binom{n-M'}{k}}{\binom{n}{k}} 
\end{align*}
for all integers $n,k$ such that $\ceil{\frac{n}{k}}=M$. 
\end{proof}

These simulation results are analogous to the known rounding results for channel coding~\cite{barman2017algorithmic}. However, while even the rounding factors are similar (although slightly stronger here!), we emphasize that the techniques needed are fundamentally different. In fact, the simulation meta-converse is related to the smoothed max-divergence \cite{fang2019quantum}, while the coding PPV meta-converse can be phrased in terms of the hypothesis testing divergence~\cite{matthews2012linear}.


\section{Non signaling assistance in the quantum setting}\label{sec:NS}

In this section we describe the non-signaling assistance of Section \ref{sec:cl-ns}  in the general quantum setting. 
Non-signaling assistance represents shared correlations that do not permit communication without external help. Let $\cW_{A_i\rightarrow B_o}$ be the quantum channel to be simulated and let $A_o, B_i$ be quantum systems of dimension $M$ where $M$ is the communication size. More precisely, a non signaling super map $\Pi$ is a super channel of input systems $A_i, B_i$ and output systems $A_o, B_o$ that maps a channel $\cW_{A_o\rightarrow B_i}$ to a channel $(\Pi\circ \cW)_{A_i \rightarrow B_o}$.  The Choi matrix of a valid non-signaling super channel $\Pi$  satisfies~\cite{fang2019quantum}: 
\begin{align*}
	&  J_\Pi \mge 0, \quad \ptr{A_oB_o}{J_\Pi}= \dI_{A_iB_i} && \text{(CP), (TP)},
	\\&  \ptr{A_o}{J_\Pi}= \frac{\dI_{A_i}}{|A_i|} \otimes \ptr{A_oA_i}{J_\Pi} && (A\nrightarrow B),
	\\&  \ptr{B_o}{J_\Pi}= \frac{\dI_{B_i}}{|B_i|} \otimes \ptr{B_iB_o}{J_\Pi}&& (B\nrightarrow A).
\end{align*}
We denote the set of non-signaling superchannels by $\cN\cS(A_oB_o|A_iB_i)$. 
The importance of non-signaling assistance lies in its ability to simplify the computation of the maximum success probability, which can be obtained by solving the following program:
\begin{align}\label{quantum-ns-program-general}
\begin{aligned}
    \suc^{\rm{NS}}(\cW,M) = \max_{\Pi} &\;1- \left\| \widetilde{\cW}_{A_i \rightarrow B_o} - \cW_{A_i \rightarrow B_o}\right\|_{\diamond} 
	\\  \rm{s.t.}\;\;& \Pi\in \cN\cS(A_oB_o|A_iB_i),
   \\ &   \widetilde{\cW}_{A_i \rightarrow B_o}=  \Pi\circ \id_{A_o \rightarrow B_i}. 
\end{aligned}
\end{align}
So, for simulating quantum channels with non-signaling strategies, the success probability can be expressed as an SDP program. Using \cite[Corollary~$2$]{fang2019quantum}, this program can be further simplified to:
\begin{align}\label{ns-program-qq}
\begin{aligned}
   \suc_{}^{\rm{NS}}(\cW, M) =\max_{\widetilde{\cW}_{A\rightarrow B}, V_B} & 1-\left\| \cW_{A\rightarrow B} - \widetilde{\cW}_{A\rightarrow B}  \right\|_{\diamond}  
	\\\rm{s.t.}\quad &  \widetilde{\cW}_{A\rightarrow B} \quad \text{quantum  channel,}  
	\\&   J_{\widetilde{\cW}} \mle  \dI_{A'} \otimes V_B, \quad  
	\\&    \tr{V_B}=M^2. 
\end{aligned}
\end{align}
Notably, when only classical instead of quantum communication is allowed, with the classical identity channel
\begin{align*}
	\id_c(\rho)=\id^c_{A_o \rightarrow B_i}(\rho)=   \sum_{x=1}^N \bra{x}\rho\ket{x} \proj{x},
\end{align*}
the non-signaling program is similar to \eqref{ns-program-qq} except that $M^2$ is replaced by $N$ (see Lemma~\ref{app:ns-program-qc-class} for a proof). Since EA strategies are non-signaling, the success probability of NS strategies is at least as its EA counterparts
\begin{align*}
	\suc^{\rm{NS}}({\cW},M) \ge \suc^{\rm{EA}} ({\cW},M). 
\end{align*}
In the following we study the gap between these success probabilties in the case of simulating classical-quantum and quantum channels. 

 
\section{Simulation of classical-quantum channels}\label{sec:CQ}

We are able to generalize the classical results of Section \ref{sec:classical} to the classical-quantum setting. In this section, we  provide rounding protocols to the simulation of classical-quantum channels with classical communication
and with entanglement-assisted strategies. 

\subsection{Entanglement-assistance}

Classical-quantum (CQ) channels can be seen as special cases of quantum channels where the input consists of commuting states. Since the output is quantum, it is natural to consider entanglement assisted strategies. The difference between SR and EA for CQ channel simulation remains open, even when considering asymptotic capacities. So, we focus on EA and define the task of simulating a CQ channel $\cW_{X\rightarrow B}$, with classical input system $X$ and quantum output system $B$, as follows.

A size-$M$ EA $\eps$-simulation code for $\cW_{X\rightarrow B}$ is a triple $(\{E^m_x \}_{m\in [M]}, \{\cN^m_{K\rightarrow B}\}_{m\in [M]}, \sigma_{K'K}  )$ consisting of a measurement POVM, a family of quantum decoding  channels, and a shared entanglement state   such that the channel:
\begin{align}\label{def:task-cq}
\widetilde{\cW}_{X\rightarrow B} : x\mapsto \sum_{m\in [M]} \cN^m_{K \rightarrow B} \left(\ptr{K'}{E^m_x \otimes \dI_{K} \cdot {\sigma}_{K'K} }   \right)
\end{align}
approximates to desired channel $\cW_{X\rightarrow B}$ in the diamond distance, i.e. it satisfies:
\begin{align*}
    \sup_{x\in \cX} \left\| \widetilde{\cW}_{X\rightarrow B}(x)- \cW_{X\rightarrow B}(x)\right\|_{\operatorname{Tr}} \le \eps.
\end{align*}
By considering a probability distribution $p_X$ as classical state $\rho_X= \sum_{x\in \cX} p_X(x) \proj{x}$, we can see that the task of simulating CQ channels is the same as the task of simulating quantum channels with classical input and  classical communication. Hence the non-signaling program is an SDP and follows from Lemma \ref{app:ns-program-qc-class}:
\begin{align}\label{ns-program-cq}
\begin{aligned}
    	\suc_{}^{\rm{NS}}(\cW, M) =\max_{\widetilde{\cW}_{X\rightarrow B}, V_B} & 1-\sup_{x\in \cX}\left\| \cW_{X\rightarrow B}(x) - \widetilde{\cW}_{X\rightarrow B}(x)  \right\|_{\operatorname{Tr}}  
	\\\rm{s.t.}\quad &  \widetilde{\cW}_{X\rightarrow B} \quad \text{CQ  channel,}  
	\\&  \widetilde{\cW}_{X\rightarrow B}(x) \mle V_B\quad  \forall x\in \cX, \quad 
	\\&    \tr{V_B}=M,  
\end{aligned}
\end{align}
where we used the equivalence
\begin{align*}
     J_{\widetilde{\cW}} \mle  \dI_{X'} \otimes V_B \Leftrightarrow \sum_{x\in \cX} \proj{x}_{X'}\otimes \widetilde{\cW}_{X\rightarrow B}(x) \mle  \dI_{X'} \otimes V_B 
     \Leftrightarrow \forall x\in \cX : \widetilde{\cW}_{X\rightarrow B}(x) \mle V_B.
\end{align*}
\subsection{Rounding}
Since entanglement assisted strategies are non-signaling, the non-signaling success probability is at least as its entanglement assisted counterpart:
\begin{align*}
    	\suc_{}^{\rm{NS}}(\cW, M) \ge 	\suc_{}^{\rm{EA}}(\cW, M).
\end{align*}
We can inquire about the extent to which these quantities can differ. Importantly, the answer to this question is similar to Corollary~\ref{cor: succ-classical} although the context and the proof strategies are slightly different. We round a feasible solution of the program \eqref{ns-program-cq} to an EA strategy. Then we compare the  distance between the resulting channel of this strategy to the feasible non-signaling channel. We prove a meta-inequality between the  EA and the NS channels that can be of independent interest.

\begin{proposition}\label{prop:meta-inequality-cq}
	Let $M, M'\in \bN$ and $\widetilde{\cW}^{\rm{NS}}_{}$ be a feasible solution of the program \eqref{ns-program-cq} of size $M$. Then, for any $M'$, there exists an entanglement-assisted strategy $\widetilde{\cW}^{\rm{EA}}_{}$ of size $M'$ such that for all $x\in \cX$:
	\begin{align*}
		\widetilde{\cW}^{\rm{EA}}(x) &\mge\left[1-\left(1-\frac{1}{M}\right)^{M'}\right]	\widetilde{\cW}^{\rm{NS}}(x).
	\end{align*}
\end{proposition}

The proof is based on a quantum analog of rejection sampling technique~\cite{cao2023channel} similar to the convex split technique (e.g., \cite{anshu2017quantum}). While the convex split technique is powerful for some tasks~\cite{anshu2017quantum,cheng2023tight,cheng2023quantum}, it not clear whether it allows to prove the multiplicative approximation of Proposition~\ref{prop:meta-inequality-cq}. Moreover, an additive approximation is possible using the convex split technique. We refer to Section~\ref{sec:quantum} where we use this technique in the fully quantum setting.

\begin{proof}[Proof of Proposition \ref{prop:meta-inequality-cq}]
    Let $(\widetilde{\cW}^{\rm{NS}}_{X\rightarrow B}, V_B)$ be a feasible solution of the program \eqref{ns-program-cq}. 
Consider the following scheme inspired by \cite{cao2023channel}.

\paragraph{Shared entanglement} Denote $V_B=V_{E_B}$
  where $E_B$ is a system hold by Bob with the same dimension as $B$. Let $\ket{V}_{E_AE_B}$ be a purification of $ \frac{V_{E_B}}{M}$ where $E_A$ is a system hold by Alice with the same dimension as $E_B$ ($|E_A|=|E_B|= |B| $). $\ket{V}_{E_AE_B}^{\otimes M'}$ is the shared-entanglement between Alice and Bob. 

\paragraph{Encoding} For an input  $x\in \cX$,  we have
\begin{align*}
	\frac{V_B}{M}\mge \left(\frac{1}{M}\right)\cdot  \widetilde{\cW}^{\rm{NS}}_{X\rightarrow B}(x).  
\end{align*}
so we can write for some state $\zeta_B(x)$ that
	\begin{align*}
		 \frac{V_B}{M}= \left(1-\lambda\right) \widetilde{\cW}^{\rm{NS}}_{X\rightarrow B}(x) + \lambda \zeta_B(x),
	\end{align*}
	where $\lambda=1-\frac{1}{M}$.
 Let
 \begin{align*}
		\ket{\widetilde{V}_x}= \sqrt{1-\lambda} \ket{0}_{S} \ket{\sigma(x)}_{Q_A E_B} + \sqrt{\lambda} \ket{1}_{S} \ket{\zeta(x)}_{Q_A E_B},
	\end{align*}
 where $\ket{\sigma(x)}$ (resp. $\ket{\zeta(x)}_{Q_A E_B}$) is a purification of $\widetilde{\cW}^{\rm{NS}}_{X\rightarrow B}(x)$ (resp. $\zeta_B(x)$),  
 $Q_A$ is the environment system hold by Alice and $E_B$ is the output system hold by Bob. We have  
 \begin{align*}
		\ptr{SQ_A}{\proj{\widetilde{V}_x}}&=\left(1-\lambda\right) \widetilde{\cW}^{\rm{NS}}_{X\rightarrow B}(x) + \lambda \zeta_B(x) 
		=  \frac{V_B}{M}=\ptr{E_A}{\proj{V}_{E_AE_B}}, 
	\end{align*}
 hence by Uhlmann's theorem~\cite{uhlmann1976transition}, there is an isometry $\cI^{(x)}$ such that,
	\begin{align*}
		\cI^{(x)}_{E_A\rightarrow SQ_A}\cdot \left(\ket{V}^{\otimes M'}_{E_AE_B} \right)= \ket{\widetilde{V}_x}_{SQ_AE_B}^{\otimes M'}.
	\end{align*}
	For an input  $x\in \cX$, Alice can then apply the isometry $\cI^{(x)}_{E_A\rightarrow S Q_A }$ on her part of the shared entangled state $\ket{V}_{E_AE_B}^{\otimes M'}$. Formally, Alice applies the isometry 
 \begin{align*}
     \cI^{}_{XE_A\rightarrow XS Q_A }=\sum_{x\in \cX} \proj{x}_X\otimes\cI^{(x)}_{E_A\rightarrow S Q_A }.
 \end{align*}
On a input $x\in \cX$, the resulting  state is,
	\begin{align*}
		 \cI^{}_{XE_A\rightarrow XS Q_A } \ket{x}_{X} \otimes\ket{V}^{\otimes M'}_{E_AE_B} &=  \ket{x}_{X}\otimes \cI^{(x)}_{E_A\rightarrow S Q_A } \ket{V}^{\otimes M'}_{E_AE_B}
 =\ket{x}_{X}\otimes \ket{\widetilde{V}_x}_{SQ_AE_B}^{\otimes M'}
  \\&=\ket{x}_{X} \sum_{y \in \{0,1\}^{M'}} 
		\sqrt{1-\lambda}^{|\bar{y}|}\sqrt{\lambda}^{|y|} \ket{y}_{S^{M'}} \otimes \ket{\sigma(x)}^{\otimes |\bar{y}|}\otimes \ket{\zeta(x)}^{\otimes|y|}, 
	\end{align*} 
	where $|y|=M'-|\bar{y}| =\sum_{i} y_i$.  
 Finally, Alice measures the system $S$ and observes $\mathbf{y}=y \in \{0,1\}^{M'} $ with probability $(1-\lambda)^{|\bar{y}|}(\lambda)^{|y|}$. The probability that $\mathbf{y}$ contains at least one ‘$0$' is
	\begin{align*}
		\pr{ 0 \in \mathbf{y}}= 1- \pr{\mathbf{y}=1\cdots 1}= 1- \lambda^{M'}. 
	\end{align*}
	Let $\mathbf{m}$ be the index of the first ‘$0$' in $\mathbf{y}$ if $\mathbf{y}\neq 1\cdots 1$ and  $M'$ otherwise. Alice sends $\mathbf{m}\in [M']$ to Bob using 
    the classical identity channel of dimension $M'$.

    In terms of the measurement channel in Definition \ref{def:task-cq}, we can make the identifications: 
    \begin{align*}
     \sigma_{K'K}&=  \proj{V}_{E_AE_B}^{\otimes M'}, \; K'= E_A, \; K=E_B,
    \\    E_x^{m} &= (\cI^{(x)}_{E_A\rightarrow SQ_A})^{\dagger} \proj{m}_{S^{M'}}\cI^{(x)}_{E_A\rightarrow SQ_A}.
    \end{align*}

 \paragraph{Decoding} 
 Bob returns the $\mathbf{m}$'s (post-measurement) copy of $E_B$. 
 This state is denoted $ \widetilde{\cW}_{\mathbf{m}}^{\rm{EA}}(x)$ with
	\begin{align*}
		 \widetilde{\cW}^{\rm{EA}}(x) &=\exs{\mathbf{m}}{ \widetilde{\cW}_{\mathbf{m}}^{\rm{EA}}(x)} 
		\\&= \sum_{y\neq 1\cdots 1} \pr{\mathbf{y}=y} \widetilde{\cW}^{\rm{NS}}(x)  + \pr{\mathbf{y}=1\cdots 1} \zeta(x)
		\\&= \left(1-\lambda^{M'} \right)   \widetilde{\cW}^{\rm{NS}}_{X\rightarrow B}(x)+ \lambda^{M'}\zeta_{B}(x)
		\\&\mge \left[1-\left(1-\frac{1}{M}\right)^{M'} \right]  \widetilde{\cW}^{\rm{NS}}_{X\rightarrow B}(x).
	\end{align*}
\end{proof}

As a consequence of Proposition~\ref{prop:meta-inequality-cq}, we can control the gap between the EA and NS success probabilities under the diamond distance. We note that a similar statement is implied by Proposition~\ref{prop:meta-inequality-cq} for fidelity-based distances as well.
\begin{corollary}\label{cor:gap-succ-cq}
	Let $M,M'\ge 1$ and $\cW$ be a CQ channel.  We have
	\begin{align*}
		1&\ge \frac{\suc^{\rm{EA}}(\cW, M')}{\suc_{}^{\rm{NS}}(\cW, M)}\ge 1- \left(1-\frac{1}{M}\right)^{M'},
  \\  1&\ge \frac{\suc^{\rm{EA}}(\cW, M\ln(t))}{\suc_{}^{\rm{NS}}(\cW, M)}\ge 1-\frac{1}{t}\quad \forall t>0\;.
	\end{align*}
\end{corollary}

\begin{proof}[Proof of Corollary~\ref{cor:gap-succ-cq}]
Let $\alpha= 1-\left(1-\frac{1}{M}\right)^{M'}$ and $x\in \cX$. Let $\widetilde{\cW}^{\rm{NS}}_{}$  be a feasible solution of the program~\eqref{ns-program-cq}. Proposition~\ref{prop:meta-inequality-cq} implies the existence of a state $\zeta$ such that
\begin{align*}
	\widetilde{\cW}^{\rm{EA}}_{X\rightarrow B}(x)= \alpha  \widetilde{\cW}^{\rm{NS}}_{X\rightarrow B}(x) +(1-\alpha) \zeta.
\end{align*}
So, we have by the triangle inequality that
\begin{align*}
	 1-\left\|  \widetilde{\cW}^{\rm{EA}}_{X\rightarrow B}(x) - {\cW}_{X\rightarrow B}(x) \right\|_{\operatorname{Tr}} 
	&=1-\left\| \alpha  \widetilde{\cW}^{\rm{NS}}_{X\rightarrow B}(x) +(1-\alpha) \zeta - {\cW}^{}_{X\rightarrow B}(x) \right\|_{\operatorname{Tr}} 
	\\&\ge 1-\alpha\left\|  \widetilde{\cW}^{\rm{NS}}_{X\rightarrow B}(x) - {\cW}^{}_{X\rightarrow B}(x) \right\|_{\operatorname{Tr}}
	\\&\quad -(1-\alpha)\left\| \zeta - {\cW}^{\rm{}}_{X\rightarrow B}(x) \right\|_{\operatorname{Tr}}
	\\&\ge 1-\alpha\left\|  \widetilde{\cW}^{\rm{NS}}_{X\rightarrow B}(x) - {\cW}^{\rm{}}_{X\rightarrow B}(x) \right\|_{\operatorname{Tr}}-(1-\alpha)
	\\&= \alpha\left(1-\left\|  \widetilde{\cW}^{\rm{NS}}_{X\rightarrow B}(x)- {\cW}^{\rm{}}_{X\rightarrow B}(x)\right\|_{\operatorname{Tr}}\right).
\end{align*}
Choosing the optimal feasible solution $\widetilde{\cW}^{\rm{NS}}_{}$ of the program~\eqref{ns-program-cq} and taking the infimum over $x\in \cX$ yields the Corollary.
\end{proof}
\sloppy Similar to the classical case, when $M'=M$, the gap between $\suc^{\rm{EA}}(\cW_{}, M) $ and $\suc^{\rm{NS}}(\cW_{}, M) $ is at most $1-\frac{1}{\mathrm{e}}$ \footnote{Achieving a similar rounding constant for coding over classical-quantum channels remains an open problem \cite{fawzi2019approximation}.}.   Furthermore, by choosing   $M'=\ln(t)M$, the gap is at most $1-\frac{1}{t}$ and approaches $1$ as $t\rightarrow \infty$. This recovers the fact that  the non-signaling correlations do not help to reduce the asymptotic entanglement-assisted simulation capacity of a CQ channel~\cite{Cao22}. 
In words, the entanglement-assisted simulation capacity of a quantum channel is the asymptotic minimum rate of communication needed to simulate this channel with high success probability when shared entanglement is available.

\begin{corollary}\label{app:cq-cap-equal}
	Let $\cW_{X\rightarrow B}$ be a CQ channel. Then, we have
	\begin{align*}
		Q^{\rm{EA}}(\cW_{X\rightarrow B})= Q^{\rm{NS}}(\cW_{X\rightarrow B}).
	\end{align*}	
\end{corollary}

\begin{proof}
From Corollary~\ref{cor:gap-succ-cq} we have for all $M=2^t,M'=\ln(t) 2^t$ and $n\in \bN$ that
\begin{align*}
	\suc^{\rm{EA}}(\cW_{X\rightarrow B}^{\otimes n}, \ln(t) 2^t)  
	& \ge  \left[1-\left(1-\frac{1}{2^t}\right)^{\ln(t) 2^t}\right]\suc^{\rm{NS}}(\cW_{X\rightarrow B}^{\otimes n}, 2^t) 
	\\& \ge \left(1-\frac{1}{t}\right)\suc^{\rm{NS}}(\cW_{X\rightarrow B}^{\otimes n}, 2^t). 
\end{align*}
Suppose by contradiction  that there is $R$ such that 
\begin{align*}
	Q^{\rm{NS}}(\cW_{X\rightarrow B})<R<Q^{\rm{EA}}(\cW_{X\rightarrow B}).
\end{align*}
So, by taking $t=Rn$ we have 
\begin{align*}
	\suc^{\rm{NS}}(\cW_{X\rightarrow B}^{\otimes n}, 2^{Rn})\underset{n \rightarrow \infty}{\longrightarrow}1
\end{align*}
and as $1-\frac{1}{t}= 1-\frac{1}{Rn} \underset{n \rightarrow \infty}{\longrightarrow}1$ we deduce that 
\begin{align*}
	\suc^{\rm{EA}}(\cW_{X\rightarrow B}^{\otimes n}, \ln(Rn)2^{Rn})\underset{n \rightarrow \infty}{\longrightarrow}1.
\end{align*}
\sloppy Hence, $R=\lim_{n \rightarrow \infty}\frac{1}{n}\log(\ln(Rn)2^{Rn})\ge Q^{\rm{EA}}(\cW_{X\rightarrow B})$
which contradicts the assumption $R<Q^{\rm{EA}}(\cW_{X\rightarrow B})$.
\end{proof}
Finally, it is not clear whether the approximation ratio we obtain in Corollary~\ref{cor:gap-succ-cq} is optimal. Contrary to the classical setting where we were able to prove the tightness of the approximation at least asymptotically (see Subsection~\ref{sec:tightness}), comparing EA and NS is notoriously difficult (See, e.g., \cite{Pironio2010May,Berta2016Aug}). We refer to Appendix~\ref{app: one shot SR vs EA} for an attempt with a special classical channel. On a technical level, and compared to the classical case, we lack an equivalent of the argument of reducing the problem to deterministic decoders in the quantum setting. 


\section{Simulation of quantum channels}\label{sec:quantum}
 
In this section, we provide rounding protocols to the simulation of quantum channels with quantum communication and with entanglement-assisted strategies (EA). Note that, in contrast, the corresponding question for quantum channel coding remains open. The price we pay for this full generality is a reduction in the power of approximation, which is weaker than what we achieved with the rounding results for classical channels in Section \ref{sec:classical} and for classical-quantum channels in Section \ref{sec:CQ}. EA strategies are more natural in the quantum setting and might be more powerful for quantum channels than shared-randomness strategies (see also Appendix \ref{app: one shot SR vs EA}).\footnote{To the best of our knowledge, it is an open question whether EA asymptotically outperform SR for quantum channel simulation.}

\sloppy Formally, a size-$M$ entanglement-assisted $\eps$-simulation code for $\cW_{A\rightarrow B}$ is a triple $(\cE_{A_iE_A\rightarrow A_o } , \cD_{E_B B_i \rightarrow B_o}, \sigma_{E_AE_B}  )$ such that $|A_o|= |B_i|=M$, $A_i=A$, $B_o=B$ and the synthesized channel
\begin{align*}
	\widetilde{\cW}_{A\rightarrow B}(\rho)= 
	\cD_{E_B B_i \rightarrow B} \,\id^q_{A_o \rightarrow B_i} \,\cE_{AE_A\rightarrow A_o }(\rho_A\otimes \sigma_{E_AE_B})
\end{align*}
is $\eps$ close to the channel $\cW_{A\rightarrow B}$ in the purified distance $P$\footnote{We can translate the results of this section to the diamond distance using Fuchs–van de Graaf inequalities \cite{fuchs1999cryptographic}], see Lemma \ref{lem:FdG}.}, where $\id_q=\id^q_{A_o \rightarrow B_i}$ denotes the quantum identity channel of dimension $M$. 

For a fixed size $M$, maximizing the success probability $1-\eps$ of EA strategies is computationally hard. As discussed in Section \ref{sec:NS}, non signaling strategies provide an SDP relaxation to compute the maximum success probability that we recall in the following  program:
\begin{align}\label{ns-program-qq-rounding-section}
\begin{aligned}
   \suc_{}^{\rm{NS}}(\cW, M) =\max_{\widetilde{\cW}_{A\rightarrow B}, V_B} & 1-P\left( \cW_{A\rightarrow B},  \widetilde{\cW}_{A\rightarrow B}  \right) 
	\\\rm{s.t.}\quad &  \widetilde{\cW}_{A\rightarrow B} \quad \text{quantum  channel,}  
	\\&   J_{\widetilde{\cW}} \mle  \dI_{A'} \otimes V_B, \quad  
	\\&    \tr{V_B}=M^2. 
\end{aligned}
\end{align}


\subsection{Rounding}

Entanglement assisted strategies are non-signaling so  the success probability of NS strategies is at least as its EA counterparts
\begin{align*}
	\suc^{\rm{NS}}({\cW},M) \ge \suc^{\rm{EA}} ({\cW},M). 
\end{align*}
Here our goal is to control the gap between these success probabilities. To this end,  we round a feasible solution of the program \eqref{ns-program-qq-rounding-section} to an EA strategy, first for fixed input state (see Proposition \ref{prop:rounding-fidelity-quantum-fixed-input}) and then for arbitrary input states (see Proposition \ref{prop:rounding-fidelity-quantum}). Note that the rounding technique (quantum rejection sampling) used in proving Proposition \ref{prop:meta-inequality-cq}--regarding the approximation of the success probability for simulating classical-quantum channels--appears to be insufficient for our current rounding endeavors, as it is limited to fixed input states without reference system. To address this, we will instead employ the convex-split technique introduced in \cite{anshu2017quantum}, which permits to handle the reference system.

\begin{proposition}\label{prop:rounding-fidelity-quantum-fixed-input}
	Let $M, M'\in \bN$ and $\widetilde{\cW}_{A\rightarrow B}^{\rm{NS}}$ be a feasible solution of the program \eqref{ns-program-qq-rounding-section} of size $M$. For all input state $\proj{\phi}_{RA}$, there exists an entanglement-assisted strategy $\widetilde{\cW}_{A\rightarrow B}^{\rm{EA}}$ of size $M'$ such that 
	\begin{align*}
    F\left(\id_R\otimes  \widetilde{\cW}^{\rm{NS}}_{A\rightarrow B}(\proj{\phi}_{RA}),\; \id_R\otimes  \widetilde{\cW}^{\rm{EA}}_{A\rightarrow B}(\proj{\phi}_{RA}) \right)\ge \frac{M'^2}{M'^2+M^2}. 
\end{align*}
\end{proposition}
\begin{proof}
    
Let $(\widetilde{\cW}^{\rm{NS}}_{A\rightarrow B}, V_B)$ be a feasible solution of the program \eqref{ns-program-qq-rounding-section}. 
	\paragraph{Shared entanglement} Denote $V_B=V_{E_B}$
  where $E_B$ is a system hold by Bob with the same dimension as $B$. Let $\ket{V}_{E_AE_B}$ be a purification of $ \frac{V_{E_B}}{M^2}$ where $E_A$ is a system hold by Alice with the same dimension as $E_B$ ($|E_A|=|E_B|= |B| $). $\ket{V}_{E_AE_B}^{\otimes M'^2}$ is  the shared-entanglement between Alice and Bob.
  \paragraph{Encoding}
A bipartite input state $\proj{\phi}_{RA}$ can be written as
\begin{align*}
	\ket{\phi}_{RA}= (O_R \otimes \dI_A)\ket{w}_{RA}, \, \text{where}\,  \ket{w}_{RA} =\sum_{i=1}^{|A|}\ket{i}_{R}  \ket{i}_A
\end{align*} 
and $R$ is a reference system with the same dimension as $A$. The state $\proj{\phi}_{RA}$ is pure so $ \spr{\phi}{\phi}=1$ implying
\begin{align*}
	\tr{O_R O_R^{\dagger }} = \bra{w} ( O_R^{\dagger }O_R \otimes \dI_A)\ket{w}= \spr{\phi}{\phi}=1. 
\end{align*}
We denote the state $O_R O_R^{\dagger } = \ptr{A}{\proj{\phi}_{RA}} $ by $\rho_R$. 
\\The constraints of the program \eqref{ns-program-qq} imply
\begin{align*}
	0&\mle \id_R\otimes  \widetilde{\cW}^{\rm{NS}}_{A\rightarrow B}(\proj{\phi}_{RA})
	\\&=  \id_R\otimes  \widetilde{\cW}^{\rm{NS}}_{A\rightarrow B}( (O_R\otimes \dI_A)\cdot  \proj{w}_{RA} \cdot  (O_R^{\dagger}\otimes \dI_A))
   \\&=(O_R\otimes \dI_B)\cdot 	J_{\widetilde{\cW}^{\rm{NS}}}  \cdot  (O_R^{\dagger}\otimes \dI_B)
	\\&\mle (O_R\otimes \dI_B)\cdot\left( \dI_{R} \otimes V_B \right)  \cdot  (O_R^{\dagger}\otimes \dI_B)
	\\&=    \rho_R \otimes V_B.
\end{align*}
Since $\tr{V_B}=M^2$ and $\tr{\rho_R}=\tr{O_R O_R^{\dagger }} = 1$, the matrix $\rho_R\otimes \frac{V_B}{M^2}$ is a quantum state satisfying for all unit vectors $\ket{\phi}_{RA}$ that
\begin{align*}
	\rho_R\otimes \frac{V_B}{M^2}\mge \left(\frac{1}{M^2}\right)\cdot  \id\otimes  \widetilde{\cW}^{\rm{NS}}_{A\rightarrow B}(\proj{\phi}_{RA}).
	\end{align*}
Therefore by the convex split Lemma of \cite{anshu2017quantum} (see Lemma \ref{lem:cvx-split}) we have that 
\begin{align*}
    F(\tau_{PQ_1Q_2\cdots Q_n}\| \tau_P \otimes \sigma_{Q_1}\otimes \sigma_{Q_2}\otimes \cdots  \otimes \sigma_{Q_n})  &\ge \frac{n}{n+k}
\end{align*}
with $R=P$, $Q=B$, $n=M'^2$, $k=M^2$, 
\begin{align*}
    \rho_{PQ} &=  \id\otimes  \widetilde{\cW}^{\rm{NS}}_{A\rightarrow B}(\proj{\phi}_{RA}), \; \sigma_Q = \frac{V_B}{M^2}, \text{ and}
    \\   \tau_{PQ_1Q_2\cdots Q_n} &= \frac{1}{n}\sum_{j=1}^n \rho_{PQ_j} \otimes \sigma_{Q_1}\otimes \sigma_{Q_2}\otimes \cdots \otimes  \sigma_{Q_{j-1}}\otimes \sigma_{Q_{j+1}}\otimes \cdots \otimes \sigma_{Q_n}. 
\end{align*}
Hence by denoting $B_{[n]\setminus j } =B_1\dots B_{j-1}B_{j+1}\dots B_n$ we  get
\begin{align}\label{cvx-split-applied}
    F\bigg(\frac{1}{M'^2}\sum_{j=1}^{M'^2} ( \id\otimes  \widetilde{\cW}^{\rm{NS}}_{A\rightarrow B}(\proj{\phi}_{RA}))_{RB_j} \otimes (\tfrac{V}{M^2})^{\otimes M'^2-1}_{B_{[n]\setminus j }} ,\; \rho_R\otimes (\tfrac{V}{M^2})^{\otimes M'^2}_{B_1\dots B_n}\bigg)\ge \frac{M'^2}{M'^2+ M^2 }
\end{align}

Let $\ket{\omega}_{R Q_AE_A^j E_B^j} $ be a purification of $\id\otimes  \widetilde{\cW}^{\rm{NS}}_{A\rightarrow B}(\proj{\phi}_{RA})$, $Q_AE_A^j\simeq AB$ is  hold by Alice and $E_B^j\simeq B$ is hold by Bob. Let $S$ be a system of dimension $M'^2$. Consider the following  state 
\begin{align*}
   \ket{\Psi}_{SRQ_A E_A^{[n] } E_B^{[n]}} = \frac{1}{M'}\sum_{j=1}^{M'^2} \ket{j}_S \otimes \ket{\omega}_{RQ_AE_A^j E_B^j} \otimes \ket{V}^{\otimes M'^2-1}_{ E_A^{[n]\setminus j } E_B^{[n]\setminus j }}. 
\end{align*}
It satisfies 
\begin{align*}
    \ptr{SQ_AE_A^{[n] }}{\proj{\Psi}_{SRQ_A E_A^{[n] } E_B^{[n]}}} = \frac{1}{M'^2}\sum_{j=1}^{M'^2} ( \id\otimes  \widetilde{\cW}^{\rm{NS}}_{A\rightarrow B}(\proj{\phi}_{RA}))_{RB_j} \otimes (\tfrac{V}{M^2})^{\otimes M'^2-1}_{B_{[n]\setminus j}}
\end{align*}
 so $\ket{\Psi}$ is a purification of $\tau$. Moreover $\ket{\phi}_{RA}\otimes \ket{V}^{\otimes M'^2}_{E_A^{[n]}E_B^{[n]}}$ is a purification of $\rho_R\otimes (\tfrac{V}{M^2})^{\otimes M'^2}_{B_1\dots B_n}$. Hence by Uhlmann's theorem~\cite{uhlmann1976transition} there is an isometry $\cI_{AE_A^{[n]} \rightarrow S Q_AE_A^{[n]}}$ such that we get from \eqref{cvx-split-applied}:
 \begin{align}\label{cvx-split-applied-Uhlmann}
   F\bigg( \ket{\Psi}_{SRQ_A E_A^{[n] } E_B^{[n]}} ,\; \cI_{AE_A^{[n]} \rightarrow S Q_AE_A^{[n]}}\big(\ket{\phi}_{RA}\otimes \ket{V}^{\otimes M'^2}_{E_A^{[n]}E_B^{[n]}}\big)\bigg) \ge \frac{M'^2}{M'^2+ M^2 }.
 \end{align}
Note that $\ket{\phi}_{RA}\otimes \ket{V}^{\otimes M'^2}_{E_A^{[n]}E_B^{[n]}}$ is the input state and shared entanglement. So Alice can apply the isometry $\cI_{AE_A^{[n]} \rightarrow S Q_AE_A^{[n]}}$ on the input system $A$ as well as her part of the shared entanglement $E_A^{[n]}$. The systems $S, Q_A$ are also held by Alice. Then Alice measures the $S$ system and obtains $\mathbf{j}\in [M'^2]$. Next, she sends $\mathbf{j}\in [M'^2]$ to Bob using super dense coding \cite{Bennett92} and the quantum identity channel of dimension $M'$. 
\paragraph{Decoding} Bob, upon receiving the index $\mathbf{j}\in [M'^2]$, returns the system $E_B^{\mathbf{j}}$. Suppose that the global state before measurement is exactly 
 \begin{align*}
   \ket{\Psi}_{SRQ_A E_A^{[n] } E_B^{[n]}} = \frac{1}{M'}\sum_{j=1}^{M'^2} \ket{j}_S \otimes \ket{\omega}_{RQ_AE_A^j E_B^j} \otimes \ket{V}^{\otimes M'^2-1}_{ E_A^{[n]\setminus j } E_B^{[n]\setminus j }}. 
\end{align*}
The post measurement state conditioned on observing $\mathbf{j}\in [M'^2]$ is then
 \begin{align*}
   \ket{\Psi_{\mathbf{j}}}_{SRQ_A E_A^{[n] } E_B^{[n]}} =  \ket{\mathbf{j}}_S \otimes \ket{\omega}_{RQ_AE_A^{\mathbf{j}} E_B^{\mathbf{j}}} \otimes \ket{V}^{\otimes M'^2-1}_{ E_A^{[n]\setminus \mathbf{j} } E_B^{[n]\setminus \mathbf{j} } }. 
\end{align*}
Therefore the state on the reference system $R$ and the system $E_B^{\mathbf{j}}$ Bob returned is 
\begin{align*}
    \ptr{SQ_A E_A^{[n] } E_B^{[n]\setminus \mathbf{j} }}{\proj{\Psi_{\mathbf{j}}}_{SRQ_A E_A^{[n] } E_B^{[n]}}} = (\id\otimes  \widetilde{\cW}^{\rm{NS}}_{A\rightarrow B}(\proj{\phi}_{RA}))_{RE_B^{\mathbf{j}}}.
\end{align*}
Hence we succeeded in the simulation task. However, the global state (after Alice applied the isometry and before she performed the measurement)  is not exactly $\ket{\Psi}_{SRQ_A E_A^{[n] } E_B^{[n]}}$ and an upper bound on the error is given by \eqref{cvx-split-applied-Uhlmann}. Denote by $(\id\otimes  \widetilde{\cW}^{\rm{EA}}_{A\rightarrow B}(\proj{\phi}_{RA}))_{RE_B^{\mathbf{j}}}$ the actual state on $RE_B^{\mathbf{j}}$ after measuring $S$, observing $\mathbf{j}$ and discarding the systems $SQ_A E_A^{[n] } E_B^{[n]\setminus \mathbf{j} }$ on the 
state $\cI_{AE_A^{[n]} \rightarrow S Q_AE_A^{[n]}}\big(\ket{\phi}_{RA}\otimes \ket{V}^{\otimes M'^2}_{E_A^{[n]}E_B^{[n]}}\big)$.
 Therefore by the data processing inequality applied on the fidelity we get from \eqref{cvx-split-applied-Uhlmann}:
\begin{align*}
    F\left((\id\otimes  \widetilde{\cW}^{\rm{NS}}_{A\rightarrow B}(\proj{\phi}_{RA}))_{RE_B^{\mathbf{j}}},\; (\id\otimes  \widetilde{\cW}^{\rm{EA}}_{A\rightarrow B}(\proj{\phi}_{RA}))_{RE_B^{\mathbf{j}}} \right)\ge \frac{M'^2}{M'^2+ M^2 }
\end{align*}
and this finishes the proof.
\end{proof}

Note that in Proposition \ref{prop:rounding-fidelity-quantum-fixed-input}, for each input state, we propose an EA strategy and control its performance. By the minimax result of \cite{Cao2024Mar}, this implies the existence of an EA assisted strategy that works \textit{for all input states} with the same performance. 
\begin{proposition}\label{prop:rounding-fidelity-quantum}
	Let $M, M'\in \bN$ and $\widetilde{\cW}^{\rm{NS}}_{A\rightarrow B}$ be a feasible solution of the program \eqref{ns-program-qq-rounding-section} of size $M$. There exists an entanglement-assisted strategy $\widetilde{\cW}^{\rm{EA}}_{A\rightarrow B}$ of size $M'$ such that for all input state $\proj{\phi}_{RA}$:
	\begin{align*}
    F\left(\id_R\otimes  \widetilde{\cW}^{\rm{NS}}_{A\rightarrow B}(\proj{\phi}_{RA}),\; \id_R\otimes  \widetilde{\cW}^{\rm{EA}}_{A\rightarrow B}(\proj{\phi}_{RA}) \right)\ge \frac{M'^2}{M'^2+ M^2 }. 
\end{align*}
\end{proposition}
\begin{proof}
    The poof given in \cite{Cao2024Mar} can be applied by replacing the actual channel ${\cW}^{}_{A\rightarrow B}$ by $\widetilde{\cW}^{\rm{NS}}_{A\rightarrow B}$.
\end{proof}

As a consequence of Proposition~\ref{prop:rounding-fidelity-quantum}, we can control the gap between the EA and NS success probabilities for simulating quantum channels under the purified  distance.

\begin{corollary}\label{cor:gap-succ-qq}
	Let $M,M'\ge 1$ and $\cW_{A\rightarrow B}$ be a quantum channel. 
     Then, we have
\begin{align}
     \suc_{}^{\rm{EA}} ({\cW},M')&\ge \sqrt{\suc_{}^{\rm{NS}}({\cW},M)}\left( \sqrt{\suc_{}^{\rm{NS}}({\cW},M)}\sqrt{\frac{M'^2}{M^2+M'^2}}  -\sqrt{\frac{2M^2}{M^2+M'^2}}\right).\label{rounding-P1}
\end{align}
Moreover, we have for $M'\ge 2\sqrt{|B|}M$ that
\begin{align}
\suc_{}^{\rm{NS}}({\cW},M') &\ge \suc_{}^{\rm{EA}} ({\cW},M') \ge g(M,M')\suc_{}^{\rm{NS}}({\cW},M),\label{rounding-P}
\end{align}
where 
\begin{align*}
     g(M,M') = \left(\sqrt{\frac{M'^2}{M^2+M'^2}}  -2\sqrt{\frac{|B|M^2}{M^2+M'^2}}\right)\underset{M'\rightarrow \infty}{\longrightarrow} 1.  
\end{align*}
Finally, we have for $M'\ge M$ that 
\begin{align}
\suc_{}^{\rm{EA}} ({\cW},M') \ge \sqrt{\frac{1}{2}}\suc_{}^{\rm{NS}}({\cW},M) - \sqrt{2}\frac{M'}{M}. 
\end{align}
\end{corollary}

\begin{proof}[Proof of Corollary~\ref{cor:gap-succ-qq}]
Note that the LHS inequality of \eqref{rounding-P} is trivial as non-signaling strategies subsume entanglement assisted ones. Let $M, M'\in \bN$ and $\widetilde{\cW}^{\rm{NS}}_{}$ be a feasible solution of the program \eqref{ns-program-qq-rounding-section} of size $M$.
By Proposition \ref{prop:rounding-fidelity-quantum}, 
 there exists an entanglement-assisted strategy $\widetilde{\cW}^{\rm{EA}}_{}$ of size $M'$ such that for all input state $\proj{\phi}_{RA}$:
	\begin{align}\label{eq:LB-fid}
    F\left(\id_R\otimes  \widetilde{\cW}^{\rm{NS}}_{A\rightarrow B}(\proj{\phi}_{RA}),  \id_R\otimes  \widetilde{\cW}^{\rm{EA}}_{A\rightarrow B}(\proj{\phi}_{RA}) \right)\ge \frac{M'^2}{M'^2+ M^2 }. 
\end{align}
Denote by  
\begin{align*}
    \rho &= \id\otimes  \widetilde{\cW}_{A\rightarrow B}(\proj{\phi}_{RA}),  &   s_e &= 1-\sqrt{1-F\left(\rho,  \rho^{\rm{EA}} \right)},
  \\  \rho^{\rm{EA}} &= \id\otimes  \widetilde{\cW}^{\rm{EA}}_{A\rightarrow B}(\proj{\phi}_{RA}), & s_n &= 1-\sqrt{1-F\left(\rho,  \rho^{\rm{NS}}\right)},
  \\\rho^{\rm{NS}} &= \id\otimes  \widetilde{\cW}^{\rm{NS}}_{A\rightarrow B}(\proj{\phi}_{RA}), & \eta &= 1-F(\rho^{\rm{NS}}, \rho^{\rm{EA}}).
\end{align*}
We have by the tight triangle inequality for the  purified distance \cite{Tomamichel2016,Ramakrishnan2023Mar} (see Lemma \ref{lem:sinus}) when $F(\rho, \rho^{\rm{NS}})+ F(\rho^{\rm{NS}}, \rho^{\rm{EA}}) \ge 1 $:
\begin{align*}
    \sqrt{1-F(\rho, \rho^{\rm{EA}})}\le \sqrt{F(\rho, \rho^{\rm{NS}})}\sqrt{1-F(\rho^{\rm{NS}}, \rho^{\rm{EA}})}+\sqrt{1-F(\rho, \rho^{\rm{NS}})}\sqrt{F(\rho^{\rm{NS}}, \rho^{\rm{EA}})}.
\end{align*}
Therefore, when $\eta$ satisfies $F(\rho, \rho^{\rm{NS}})+ F(\rho^{\rm{NS}}, \rho^{\rm{EA}})\ge 1 \Leftrightarrow 1-(1-s_n)^2+1-\eta \ge 1\Leftrightarrow \eta \le 2s_n-s_n^2  $ we have that
\begin{align}
    1-s_e &\le \sqrt{1-(1-s_n)^2}\sqrt{\eta} + (1-s_n)\sqrt{1-\eta}\notag
    \\ \Rightarrow s_e &\ge (1-\sqrt{1-\eta}) +s_n \sqrt{1-\eta} -\sqrt{2s_n-s_n^2}\sqrt{\eta} \notag
    \\&\ge s_n \sqrt{1-\eta} -\sqrt{2s_n}\sqrt{\eta}
    = \sqrt{s_n}(\sqrt{s_n}\sqrt{1-\eta} - \sqrt{2\eta} ). \label{eq:se-sn-eta}
\end{align}
Note that if $\eta > 2s_n -s_n^2 \ge \frac{s_n}{2+s_n} \Rightarrow \sqrt{s_n}\sqrt{1-\eta} < \sqrt{2\eta}$ so the inequality \eqref{eq:se-sn-eta} trivially holds and we can drop the condition $\eta \le  2s_n -s_n^2$.
\\Moreover, by \eqref{eq:LB-fid}, we have $\eta = 1-F(\rho^{\rm{NS}}, \rho^{\rm{EA}}) \le \frac{M^2}{M^2+M'^2}$ so
\begin{align*}
     s_e&\ge \sqrt{s_n}(\sqrt{s_n}\sqrt{1-\eta} - \sqrt{2\eta} )
     \ge \sqrt{s_n}\left( \sqrt{s_n}\sqrt{\frac{M'^2}{M^2+M'^2}}  -\sqrt{\frac{2M^2}{M^2+M'^2}}\right). 
\end{align*}
Choosing the optimal feasible solution $\widetilde{\cW}^{\rm{NS}}_{}$ of the program~\eqref{ns-program-qq-rounding-section} and taking the infimum over ${\phi}_{RA}$ yields \eqref{rounding-P1}. 
\\Moreover since  we have always $s_n\ge 1-\sqrt{1-\frac{1}{|B|}}\ge \frac{1}{2|B|}$ we get for $M,M'$ satisfying $\frac{M^2}{M^2+M'^2}\le \frac{1}{1+4|B|} \le \frac{1}{2|B|} \le s_n \le 2s_n-s_n^2$:
\begin{align*}
     s_e&\ge \sqrt{s_n}(\sqrt{s_n}\sqrt{1-\eta} - \sqrt{2\eta} )
     \\&\ge \sqrt{s_n}(\sqrt{s_n}\sqrt{1-\eta} - \sqrt{2|B|s_n}\sqrt{2\eta} )
     \\&= (\sqrt{1-\eta} - 2\sqrt{\eta |B|} )s_n
     \\&\ge \left(\sqrt{\frac{M'^2}{M^2+M'^2}}  -2\sqrt{\frac{|B| M^2}{M^2+M'^2}}\right)s_n. 
\end{align*}
Choosing the optimal feasible solution $\widetilde{\cW}^{\rm{NS}}_{}$ of the program~\eqref{ns-program-qq-rounding-section} and taking the infimum over ${\phi}_{RA}$ yields \eqref{rounding-P} for $M'\ge 2\sqrt{|B|}M$.
\end{proof}

Similar to the classical and classical-quantum  settings, we can use Corollary \ref{cor:gap-succ-qq} to reprove the fact that  the non-signaling correlations do not help to reduce the asymptotic entanglement-assisted simulation capacity of a quantum channel~\cite{cubitt2011zero,fang2019quantum}. 
In words, the entanglement-assisted simulation capacity of a quantum channel is the asymptotic minimum rate of communication needed to simulate this channel with high success probability when shared entanglement is available.

\begin{corollary}\label{app:qu-cap-equal}
	Let $\cW_{A\rightarrow B}$ be a quantum channel. Then, we have
	\begin{align*}
		Q^{\rm{EA}}(\cW_{A\rightarrow B})= Q^{\rm{NS}}(\cW_{A\rightarrow B}).
	\end{align*}	
\end{corollary}

\begin{proof}
Let $R>0$ be a positive real number and $n \in \mathbb{N}$. 
From Corollary~\ref{cor:gap-succ-qq} Eq.~\eqref{rounding-P1} we have for all $M=2^{nR},M'=n 2^{nR}$  that
\begin{align}\label{eq:succ-P-EA-NS}
	1\ge &\suc^{\rm{EA}}(\cW_{A\rightarrow B}^{\otimes n}, n 2^{nR})  \notag 
	\\& \qquad \ge \sqrt{\suc^{\rm{NS}}(\cW_{A\rightarrow B}^{\otimes n},  2^{nR})}  \left(\sqrt{\suc^{\rm{NS}}(\cW_{A\rightarrow B}^{\otimes n},  2^{nR})}\sqrt{\frac{n^2}{1+n^2}} - \sqrt{\frac{2}{1+n^2}}\right). 
\end{align}
Suppose by contradiction  that there is $R$ such that 
\begin{align*}
	Q^{\rm{NS}}(\cW_{A\rightarrow B})<R<Q^{\rm{EA}}(\cW_{A\rightarrow B}).
\end{align*}
So, we have 
\begin{align*}
	\suc^{\rm{NS}}(\cW_{A\rightarrow B}^{\otimes n}, 2^{nR})\underset{n \rightarrow \infty}{\longrightarrow}1
\end{align*}
and as $\sqrt{\frac{n^2}{1+n^2}} \underset{n \rightarrow \infty}{\longrightarrow}1$ and $\sqrt{\frac{2}{1+n^2}} \underset{n \rightarrow \infty}{\longrightarrow} 0$ we deduce from \eqref{eq:succ-P-EA-NS} that 
\begin{align*}
	\suc^{\rm{EA}}(\cW_{A\rightarrow B}^{\otimes n}, n2^{nR})\underset{n \rightarrow \infty}{\longrightarrow}1.
\end{align*}
Hence, $R=\lim_{n \rightarrow \infty}\frac{1}{n}\log(n2^{nR})\ge Q^{\rm{EA}}(\cW_{A\rightarrow B})$
which contradicts the assumption $R<Q^{\rm{EA}}(\cW_{A\rightarrow B})$.
\end{proof}

Finally, it is likely that the approximation ratio we obtain in Corollary~\ref{cor:gap-succ-qq} is not optimal and we leave it as an open question to improve this rounding approximation so that it matches its classical  (see Corollary~\ref{cor: succ-classical}) and  classical-quantum  (see Corollary~\ref{cor:gap-succ-cq}) counterparts.


\section{Conclusion}

In this paper, we studied approximation algorithms for simulating classical, classical-quantum, and quantum channels. We provided rounding results relating the success probabilities of shared-randomness assistance and entanglement-assistance to non-signaling assisted strategies. In particular, we deduced that non-signaling assistance is not proving any further advantage for asymptotic i.i.d.\ simulation capacities. Yet, non-signaling assisted values are linear or semi-definite programs and are therefore computationally more tractable for approximating the optimal success probabilities. We further proved that our rounding results are tight in the classical case with respect to shared-randomness assistance, but had to leave open the tightness question when comparing entanglement-assistance with non-signaling assistance (cf.~\textbf{}Appendix~\ref{app: one shot SR vs EA}). This relates to the notoriously difficult problem of bounding the power of entanglement in a more refined way (see, e.g., \cite{Pironio2010May,Berta2016Aug}).

It would be interesting to study the impact of our one-shot result on more refined information-theoretic asymptotic i.i.d.\ limits, such as error exponents or strong converse exponents (which are largely unknown for the task of channel simulation). Further generalizations to network topologies might also be possible. Finally, investigating computational hardness results of achieving an improved approximation ratio of $(1-\mathrm{e}^{-1}+\eps)$ for some $\eps>0$ also remains an open question. 

	
\section*{Acknowledgment}

MB and AO acknowledge funding by the European Research Council (ERC Grant Agreement No. 948139), MB acknowledges support from the Excellence Cluster - Matter and Light for Quantum Computing (ML4Q), and OF acknowledges support from the European Research Council (ERC Grant AlgoQIP, Agreement No. 851716).

\printbibliography


\appendix

\section{Deferred proofs}

We restate and prove the symmetry based reduction of the non-signaling program of simulating a classical channel~\eqref{classical-ns-program-sym}:
\begin{lemma}\label{app:ns-program-cc}
	The non-signaling program of simulating a classical channel can be written as
	\begin{align}
		\suc^{\rm{NS}} = \max_{\widetilde{W}_{}, \zeta} & 1-\sup_{x\in \cX}\left\|\widetilde{W}(\cdot|x)_{}- W(\cdot|x)_{}\right\|_{\TV} \notag 
		\\ \rm{s.t.}\;\;&  \sum_{y} \widetilde{W}_{}(y|x)= 1 \quad \forall x, \notag
		\\&   \widetilde{W}_{}(y|x)\ge 0 \quad \quad \;\;\, \forall x,y, \notag 
		\\&  \widetilde{W}_{}(y|x) \le \zeta(y)  \quad\;\, \forall x,y, \notag
		\\&  \sum_{y} \zeta(y)=M.  \notag
	\end{align}  
\end{lemma}

\begin{proof}
Recall the non-signaling program of simulating a classical channel
\begin{align}
	\suc^{\rm{NS}} = \max_{N_{IY|XJ}} &1-\sup_{x\in \cX}\left\|\widetilde{W}(\cdot|x)_{}- W(\cdot|x)_{}\right\|_{\TV}\notag
	\\  \rm{s.t.}\;\;& \widetilde{W}(y|x)= \sum_{i=1}^M N_{IY|XJ}(i,y|x,i) \notag 
	\\&   N_{IY|XJ}(i,y|x,j)\ge 0 \quad \forall i,y,x,j,\notag
	\\&\sum_{i,y}    N_{IY|XJ}(i,y|x,j)=1  \quad \forall x,j,\notag 
	\\    &\sum_{i}    N_{IY|XJ}(i,y|x,j)=N_{Y|J}(y|j)  \quad \forall x, \notag
	\\    &\sum_{y}    N_{IY|XJ}(i,y|x,j)=N_{I|X}(i|x)  \quad  \forall j.\notag
\end{align}
Given a non-signaling super map 	$N_{IY|XJ}$, we define $\zeta$ as
\begin{align*}
	\forall y : \zeta(y)= \sum_{j}N_{Y|J}(y|j).
\end{align*}
We have by using the non-signaling properties of $N_{IY|XJ}$ that
\begin{align*}
	\forall x,y : \zeta(y)&= 
	\sum_{j}N_{Y|J}(y|j)
		\\&=\sum_{i,j}N_{IY|XJ}(i,y|x,j)
	\\	&\ge \sum_{i}N_{IY|XJ}(i,y|x,i)
	\\&= \widetilde{W}(y|x),
	\\	\sum_{y}\zeta(y)&= \sum_y\sum_{i,j}N_{IY|XJ}(i,y|x,i)
	\\&=\sum_{j}\sum_{i,y}N_{IY|XJ}(i,y|x,j)
	\\&= \sum_j 1= M.
\end{align*}
Conversely if $\zeta$ satisfies
\begin{align*}
	\widetilde{W}(y|x) &\le \zeta(y)  \quad\;\, \forall x,y, 
	\\  \sum_{y} \zeta(y)&=M,
\end{align*}
we define the non-signaling super map $N_{IY|XJ}$ as
\begin{align*}
	N_{IY|XJ}(i,y|x,j)=
\begin{cases}
	\frac{\widetilde{W}(y|x)}{M} 	& \text{if}\; i=j,
	\\ \frac{\zeta(y)-\widetilde{W}(y|x)}{M(M-1)}& \text{if}\; i\neq j.
\end{cases}
\end{align*}
We check the non-signaling properties
\begin{align*}
	N_{IY|XJ}(i,y|x,j)&\ge 0,
	\\\sum_{i} N_{IY|XJ}(i,y|x,j)&=\frac{\zeta(y)}{M},
	\\\sum_{y} N_{IY|XJ}(i,y|x,j)&=	\frac{1}{M},
	\\	\sum_{i,y} N_{IY|XJ}(i,y|x,j)&=\sum_i \frac{1}{M}=1.
\end{align*}
Hence, the two programs are equal.
\end{proof}

Similarly, we prove the symmetry based reduction of the non-signaling program of simulating a quantum channel with classical communication. Note that this program is similar to the one with quantum communication~\eqref{ns-program-qq}.

\begin{lemma}\label{app:ns-program-qc-class}
	The non-signaling program of simulating a quantum channel with classical communication of size $N$ can be written as
	\begin{align}
		\suc_{C}^{\rm{NS}} =\max_{\widetilde{\cW}_{A\rightarrow B}, V_B} & 1-\left\| \cW_{A\rightarrow B}- \widetilde{\cW}_{A\rightarrow B}  \right\|_{\diamond} \notag 
		\\\rm{s.t.}\quad &  \widetilde{\cW}_{A\rightarrow B} \quad \text{quantum  channel,}   \label{ns-program-qq-cl}
		\\&   J_{\widetilde{W}} \mle  \dI_{A'} \otimes V_B, \quad  \notag
		\\&    \tr{V_B}=N.  \notag
	\end{align}
\end{lemma}

\begin{proof}
The general form of the non-signaling program of simulating a quantum channel $\cW$ using a quantum channel $\cN$ is~\cite{fang2019quantum}
\begin{align}\label{ns-program-general-form}
	\suc_{\cN}^{\rm{NS}} =\max_{\widetilde{\cW}_{A\rightarrow B }} & 1-\left\| \cW_{A\rightarrow B}- \widetilde{\cW}_{A\rightarrow B}  \right\|_{\diamond}  \notag
	\\\rm{s.t.}\quad &  \widetilde{\cW}_{A\rightarrow B} \quad \text{quantum  channel,}\notag  
	\\&   J_{\widetilde{\cW}}= \ptr{A_oB_i}{(J_{\cN}^\top \otimes \dI_{A_iB_o})\cdot J_\Pi },\notag 
	\\&  J_\Pi \ge 0, \quad \ptr{A_oB_o}{J_\Pi}= \dI_{A_iB_i},
	\\&  \ptr{A_o}{J_\Pi}= \frac{\dI_{A_i}}{|A_i|} \otimes \ptr{A_oA_i}{J_\Pi},\notag 
	\\&  \ptr{B_o}{J_\Pi}= \frac{\dI_{B_i}}{|B_i|} \otimes \ptr{B_iB_o}{J_\Pi}. \notag 
\end{align}
The setting of simulation using classical communication corresponds to $\cN=\id_c(\rho)= \sum_{x=1}^N \bra{x}\rho\ket{x} \proj{x} $. Using the symmetry of $J_{\id_c}$, we can conjugate the Choi matrix $J_\Pi$ with the unitaries for $\sigma \in \fS_N$:
\begin{align*}
	P^\sigma_{A_o}\otimes P^\sigma_{B_i}= \sum_{x=1}^N\ket{x}\bra{\sigma(x)}_{A_o} \otimes \sum_{x=1}^N\ket{x}\bra{\sigma(x)}_{B_i}.
\end{align*}
If $J_\Pi$ is an  optimal solution of the program \eqref{ns-program-general-form}  then  for $\sigma\in \fS_N$, the   Choi matrix $\left( P^\sigma_{A_o}\otimes P^\sigma_{B_i}\right)\cdot  J_\Pi \cdot \left( P^{\sigma^{-1}}_{A_o}\otimes P^{\sigma^{-1}}_{B_i}\right)$ satisfies the constraints of the program \eqref{ns-program-general-form}  hence $\frac{1}{N!}\sum_{\sigma\in \fS_N} \left( P^\sigma_{A_o}\otimes P^\sigma_{B_i}\right)\cdot  J_\Pi \cdot \left( P^{\sigma^{-1}}_{A_o}\otimes P^{\sigma^{-1}}_{B_i}\right) $ is also an optimal solution. 
Hence, we can write
\begin{align*}
	J_\Pi &= \sum_{x} \ket{x}\bra{x}_{A_o }  \otimes \ket{x}\bra{x}_{ B_i} \otimes  \rho^{0,0,0,0}_{A_iB_o} 
	\\&\quad + \sum_{x\neq y} \ket{x}\bra{x}_{A_o }  \otimes \ket{y}\bra{y}_{ B_i} \otimes  \rho^{1,1,0,0}_{A_iB_o}  
	\\&\quad+\qquad \cdots 
	\\&\quad+ \sum_{x\neq y\neq z \neq t} \ket{x}\bra{y}_{A_o }  \otimes \ket{z}\bra{t}_{ B_i} \otimes  \rho^{3,2,1,0}_{A_iB_o}.
\end{align*}
Denote by $H_{A_iB_0}= \rho_{A_iB_0}^{0,0,0,0}$ and $K_{A_iB_0}= \rho_{A_iB_o}^{1,1,0,0}$ then we have
\begin{align*}
	J_{\widetilde{\cW}}&= \sum_x \bra{xx}\otimes \dI_{A_iB_o} \; J_\Pi \;\ket{xx}\otimes \dI_{A_iB_o}   
	= N \rho_{A_iB_0}^{0,0,0,0}=NH_{A_iB_0} . 
\end{align*}
Moreover, the constraint $\ptr{A_o}{J_\Pi}= \frac{\dI_{A_i}}{|A_i|} \otimes \ptr{A_oA_i}{J_\Pi}$ implies
\begin{align*}
	&  N\left(\rho^{0,0,0,0}_{A_iB_o}+ (N-1)\rho^{1,1,0,0}_{A_iB_o}\right)
	=\frac{\dI_{A_i}}{|A_i|} \otimes N\left( \rho^{0,0,0,0}_{B_o}+ (N-1)\rho^{1,1,0,0}_{B_o}\right).       
\end{align*}
Denoting
\begin{align*}
	V&= \frac{N}{|A_i|} \left( \rho^{0,0,0,0}_{B_o}+ (N-1)\rho^{1,1,0,0}_{B_o}\right)
	= \frac{N}{|A_i|} \left( H_{B_o}+ (N-1)K_{B_o}\right),
\end{align*}
we have 
\begin{align*}
	\dI_{A_i}\otimes V&= \dI_{A_i} \otimes \frac{N}{|A_i|}\left( \rho^{0,0,0,0}_{B_o}+ (N-1)\rho^{1,1,0,0}_{B_o}\right)
	=N\left(\rho^{0,0,0,0}_{A_iB_o}+ (N-1)\rho^{1,1,0,0}_{A_iB_o}\right).
\end{align*}
On the other hand, we have $J_\Pi\mge 0$ so $K\mge 0$, and hence 
\begin{align*}
	J_{\widetilde{\cW}} &= NH_{A_iB_0} \mle N(H_{A_iB_0}+ (N-1) K_{A_iB_0} ).
\end{align*}
Since $H_{A_iB_0}=\rho^{0,0,0,0}_{A_iB_o}$ and $K_{A_iB_0}=\rho^{1,1,0,0}_{A_iB_o}$ it follows that
\begin{align*}
	J_{\widetilde{\cW}} &\mle N\left(\rho^{0,0,0,0}_{A_iB_0}+ (N-1) \rho^{1,1,0,0}_{A_iB_0} \right)
	= \dI_{A_i}\otimes V.
\end{align*}
Moreover, we have
\begin{align*}
	|A_i|\cdot |B_i|&= \tr{\dI_{A_iB_i}}=\tr{J_\Pi}
			=N\tr{H}+N(N-1)\tr{K}= |A_i| \cdot \tr{V}
\end{align*}
and consequently
\begin{align*}
	\tr{V}= |B_i|=N.
\end{align*}
Therefore, we find $\suc_{\id_c}^{\rm{NS}}\le \rm{Opt} \;\eqref{ns-program-qq-cl}$. 
		
Conversely, given $V_{B_o}$ such that $  J_{\widetilde{\cW}}\mle\dI_{A_i}\otimes V_{B_o}$ and $\tr{V_{B_o}}=N$,   we can  choose the Choi matrix of the non-signaling strategy as
\begin{align*}
	J_\Pi=& \frac{1}{N}\sum_{x} \proj{x}_{A_0}\otimes \proj{x}_{B_i}\otimes J_{\widetilde{\cW}}
	+& \frac{1}{N(N-1)}\sum_{x\neq z} \proj{x}_{A_0}\otimes \proj{z}_{B_i}\otimes \left(\dI_{A_i}\otimes V_{B_o} - J_{\widetilde{\cW}} \right).
\end{align*}
Both $\Pi$ and $\widetilde{\cW}$ satisfy all the constraints of the program \eqref{ns-program-general-form} for $\cN= \id_c$, hence $\rm{Opt} \;\eqref{ns-program-qq-cl}\le \suc_{\id_c}^{\rm{NS}}$, and finally
\begin{align*}
	\rm{Opt} \;\eqref{ns-program-qq-cl}= \suc_{\id_c}^{\rm{NS}}.
\end{align*}
\end{proof}


\section{One shot SR, EA and NS  success probabilities}\label{app: one shot SR vs EA}

Here we follow \cite{Berta2016Aug} and show that the SR, EA and NS success probabilities are in general different (see the discussion at the end of Subsection~\ref{sec:cl-ns}). Consider the channel \cite{cubitt2011zero,Berta2016Aug}, which is a special type of the channels~\eqref{eq:example-channel} we used to prove the tightness of our rounding results:
\begin{align*}
W : \binom{4}{2}\rightarrow \{1,2, 3, 4\}, \; W(y|x)=\frac{1}{2}\;\mathbf{1}\{y \in x\}.
\end{align*}
We have from Lemma~\ref{app:optimal} $\suc^{\rm{NS}}(W, 2)=1$ and $\suc^{\rm{SR}}(W, 2)\le \frac{5}{6}$. We use the same techniques of \cite{Berta2016Aug} to show that 
\[\suc^{\rm{EA}}(W, 2)\gtrsim 0.908. \]
Define the shared quantum state
\[\ket{\psi}= \frac{1}{2}\sum_{i=1}^4 \ket{i}\otimes \ket{i}. \]
Upon receiving the message $m\in\{1,2\}$, the decoder performs a measurement  with  the computational basis if $m=1$ and otherwise in the  orthonormal basis corresponding to $U$ with 
\[U=\frac{1}{\sqrt{3}}\begin{pmatrix}
   0&-1&-1&1\\
   1&0&1&1 \\
   -1&1&0&1\\
   -1&-1&1&0
\end{pmatrix}. \]
In other words, we set $D(y|1)= \proj{y}$ and $D(y|2)=U\proj{y}U^\dagger$ for $y\in \{1,2,3,4\}$. For an input $x\in \binom{4}{2}$, the encoder performs the measurement $\{E(m|x)\}_{m\in \{0,1\}}$ on the first system of the shared state. The EA success probability satisfies
\begin{align*}
    \suc^{\rm{EA}}(W, 2)&\ge \inf_{x\in \binom{4}{2}} \sum_{y} \min\left(\widetilde{W}^{\rm{EA}}(y|x),  \frac{1}{2}\;\mathbf{1}\{y \in x\}\right)
    \\&= \inf_{x\in \binom{4}{2}} \sum_{y\in x} \min\left(\widetilde{W}^{\rm{EA}}(y|x),  \frac{1}{2}\right).
\end{align*}
By the definition and description  of the entangled  strategy above, we have that
\begin{align*}
    \widetilde{W}^{\rm{EA}}(y|x) &= \sum_{m=1}^2\bra{\psi}E(m|x)\otimes D(y|m)\ket{\psi} 
    \\&= \bra{\psi}E(1|x)\otimes D(y|1)\ket{\psi} + \bra{\psi}[\dI-E(1|x)]\otimes D(y|2)\ket{\psi} 
    \\&= \bra{\psi}E(1|x)\otimes (\proj{y}-U\proj{y}U^\dagger)\ket{\psi}  + \bra{\psi}\dI\otimes U\proj{y}U^\dagger\ket{\psi} 
       \\&=\frac{1}{4} \tr{E(1|x)^\top\cdot (\proj{y}-U\proj{y}U^\dagger)} + \frac{1}{4}\le \frac{1}{2}
\end{align*}
using the fact that $0\mle E(1|x)\mle \dI$. Hence, we find
\begin{align*}
    \suc^{\rm{EA}}(W, 2)&\ge \inf_{x\in \binom{4}{2}} \sum_{y\in x} \min\left(\widetilde{W}^{\rm{EA}}(y|x),  \frac{1}{2}\right)
    \\&= \inf_{x\in \binom{4}{2}} \sum_{y\in x}\frac{1}{4} \tr{E(1|x)^\top\cdot (\proj{y}-U\proj{y}U^\dagger)} + \frac{1}{4}
    \\&= \frac{1}{2}+ \frac{1}{4}\inf_{x\in \binom{4}{2}} \tr{E(1|x)^\top\cdot \left(\sum_{y\in x}\proj{y}-U\sum_{y\in x}\proj{y}U^\dagger\right)}
      \\&= \frac{1}{2}+ \frac{1}{4}\inf_{x\in \binom{4}{2}} \frac{1}{2} \left\|\sum_{y\in x}\proj{y}-U\sum_{y\in x}\proj{y}U^\dagger\right\|_1,
\end{align*}
where we maximize over $\{E(1|x)\}_x$ in the last step. Since  we have that $\left\|\sum_{y\in x}\proj{y}-U\sum_{y\in x}\proj{y}U^\dagger\right\|_1 \simeq 3.2660$ for all $x\in \binom{4}{2}$, we deduce that 
\begin{align*}
     \suc^{\rm{EA}}(W, 2)&\gtrsim 0.908 > \frac{5}{6}. 
\end{align*}
On the other hand, observe that we have for $\widetilde{W}^{\rm{EA}}(y|x) = \sum_{m=1}^2\bra{\psi}E(m|x)\otimes D(y|m)\ket{\psi} $:
\begin{align*}
     \suc^{\rm{EA}}(W, 2) &= \sup_{\ket{\psi}, E, D } \inf_{x\in \binom{4}{2}} \sum_{y\in x} \min\left(\widetilde{W}^{\rm{EA}}(y|x),  \frac{1}{2}\right)
     \\&\le \sup_{\ket{\psi}, E, D } \frac{1}{\binom{4}{2}}\sum_{x\in \binom{4}{2}} \sum_{y\in x} \widetilde{W}^{\rm{EA}}(y|x)
     \\&= \sup_{\ket{\psi}, E, D } \frac{1}{2}\sum_{m=1}^2  \sum_{x\in \binom{4}{2}} \sum_{y\in [4]}   \frac{1}{3}\bid\{y\in x\} \bra{\psi}E(m|x)\otimes D(y|m)\ket{\psi}. 
\end{align*}
The latter expression can be identified as the optimal entanglement assisted coding success probability for the channel $N(x|y)= \frac{1}{3}\bid\{x \ni y\}$ where the encoder is represented by $D$ and the decoder is represented by $E$. In \cite{Berta2016Aug}, it is shown that this optimal success probability can be upper bounded  by approximately $0.908$ and thus 
\begin{align*}
      \suc^{\rm{EA}}(W, 2) &\lesssim 0.908. 
\end{align*}
Finally, we find that
\begin{align*}
    \suc^{\rm{SR}}(W, 2) \lneqq \suc^{\rm{EA}}(W, 2) \lneqq  \suc^{\rm{NS}}(W, 2). 
\end{align*}


\section{Lemmas}

\begin{lemma}[\cite{anshu2017quantum}, Lemma 2.1]\label{lem:cvx-split}
    Let $\rho_{PQ} \in\mathrm{D}(PQ)$ and $\sigma_Q \in \mathrm{D}(Q)$ be quantum states such that $\mathrm{supp}(\rho_Q)\subset \mathrm{supp}(\sigma_Q)$. Let $ k = \inf\{\lambda\in \mathbb{R} : \lambda \rho_P \otimes \sigma_Q \mge \rho_{PQ}\}$. Define the following state 
    \begin{align*}
        \tau_{PQ_1Q_2\cdots Q_n} = \frac{1}{n}\sum_{j=1}^n \rho_{PQ_j} \otimes \sigma_{Q_1}\otimes \sigma_{Q_2}\otimes \cdots \otimes  \sigma_{Q_{j-1}}\otimes \sigma_{Q_{j+1}}\otimes \cdots \otimes \sigma_{Q_n}
    \end{align*}
    on $(n+1)$ registers $P, Q_1, Q_2, \dots, Q_n$, where $\forall j \in [n] : \rho_{PQ_j} = \rho_{PQ}$ and $\sigma_{Q_j} = \sigma_Q$. Then, 
    \begin{align*}
        D(\tau_{PQ_1Q_2\cdots Q_n}\| \tau_P \otimes \sigma_{Q_1}\otimes \sigma_{Q_2}\otimes \cdots  \otimes \sigma_{Q_n}) &\le \log\bigg(1+\frac{k}{n} \bigg),
        \\ F(\tau_{PQ_1Q_2\cdots Q_n}\| \tau_P \otimes \sigma_{Q_1}\otimes \sigma_{Q_2}\otimes \cdots  \otimes \sigma_{Q_n})  &\ge \frac{1}{1+\frac{k}{n}}.
    \end{align*}
\end{lemma}

\begin{lemma}[Fuchs–van de Graaf inequalities \cite{fuchs1999cryptographic}]\label{lem:FdG} 
    Let $\rho \in \mathrm{D}(A)$ and $\sigma \in \mathrm{D}(A)$ be quantum states. We have 
    \begin{align*}
      1-\sqrt{F(\rho, \sigma)} &\le \|\rho-\sigma\|_{\mathrm{tr}} \le \sqrt{1-F(\rho, \sigma)}.
    \end{align*}
\end{lemma}

\begin{lemma}[\cite{Tomamichel2016,Ramakrishnan2023Mar}]\label{lem:sinus}
Let $\rho \in \mathrm{D}(A)$, $\sigma \in \mathrm{D}(A)$ and $\tau  \in \mathrm{D}(A)$ be quantum states such that $F(\rho, \sigma)+ F(\sigma, \tau)\ge 1$. We have that 
    \begin{align*}
    \sqrt{1-F(\rho, \tau)}\le \sqrt{F(\rho, \sigma)}\sqrt{1-F(\sigma, \tau)}+\sqrt{1-F(\rho, \sigma)}\sqrt{F(\sigma, \tau)}.
\end{align*}
\end{lemma}


\end{document}